\documentclass[11pt]{article}

\usepackage{lineno,hyperref}
\modulolinenumbers[5]

\usepackage{amssymb}
\usepackage{amsfonts}
\usepackage{amsmath}
\usepackage{amsbsy}
\usepackage{graphicx}
\usepackage{enumerate}
\usepackage{geometry}
\usepackage{tikz}
\usepackage{proof}
\providecommand{\U}[1]{\protect\rule{.1in}{.1in}}

\allowdisplaybreaks
\newtheorem{theorem}{Theorem}

\newtheorem{corollary}[theorem]{Corollary}

\newtheorem{definition}[theorem]{Definition}
\newtheorem{example}[theorem]{Example}

\newtheorem{lemma}[theorem]{Lemma}

\newtheorem{proposition}[theorem]{Proposition}
\newtheorem{remark}[theorem]{Remark}

\newenvironment{proof}[1][Proof]{\noindent\textbf{#1.} }{\ \rule{0.5em}{0.5em}}
\newcommand{\hide}[1]{}
\newcommand{\cupdot}{\mathbin{\mathaccent\cdot\cup}}

\begin{document}

\title{Weighted propositional configuration logics: A specification language for architectures with quantitative features}
\author{Paulina Paraponiari$^{a}$, George Rahonis$^{b,}$\footnote{Corresponding author}\\Department of Mathematics\\Aristotle University of Thessaloniki\\54124 Thessaloniki, Greece\\$^{a}$parapavl@math.auth.gr, $^{b}$grahonis@math.auth.gr}
\date{}
\maketitle

\begin{abstract}
We introduce and investigate a weighted propositional configuration logic over  commutative semirings. Our logic is intended to serve as a specification language for software architectures with quantitative features. We prove an efficient construction of full normal forms and decidability of equivalence of formulas in this logic. We illustrate the  motivation of this work by describing well-known architectures equipped with quantitative characteristics using formulas in our logic.

\end{abstract}

\section{Introduction}
Architecture is a critical issue in design and development of  complex software systems. Whenever the construction of a software system is based on a ``good'' architecture, then the system satisfies most of its functional and quality requirements. But what are the characteristics of a ``good'' architecture and how can one design it? Despite the huge progress on software architecture, over almost three decades, the field remains relatively immature (cf. \cite{Ga:So} for a detailed presentation of the progress of software architecture). Several fundamental matters still remain open, for instance the distinction between architectures and their properties. 
To address these issues there is a need of a formal treatment of architectures. The very recent work of \cite{Ma:Co} contributes towards this direction by studying the relation among architectures and architecture styles. An architecture consists of a fixed number of components connected in a concrete topology. An architecture style describes a family of ``similar'' architectures, i.e., architectures with the same types of components and topologies. The authors of \cite{Ma:Co} introduced the propositional configuration logic (PCL for short) which was proved sufficient to describe architectures: the meaning of every PCL formula is a configuration set, which intuitively represents permissible component connections, and every architecture can be represented by a configuration set on the collection of its components. The first-order and second-order extensions of PCL described perfectly the concept of architecture styles. Therefore, PCL and its first- and second-order extensions constitute logics for the specification of architecture styles and hence, an important contribution to rigorous systems design (cf. \cite{Si:Ri}).  

In this paper we extend the work of \cite{Ma:Co} in a quantitative setup, and specifically we introduce and investigate a weighted PCL (wPCL for short) over a commutative semiring $(K, \oplus, \otimes, 0,  1)$. Our work is motivated as follows. Propositional configuration logic  of \cite{Ma:Co} describes in a strict mathematical way qualitative features of architectures. Nevertheless, in several practical applications also quantitative characteristics of architectures are required: the costs of the interactions among the components of an architecture, the time needed, the probability of the implementation of a concrete interaction, etc.  For instance, the development of IoT and cloud applications based on the well-known Publish/Subscribe architecture \cite{Pa:Pu,Ol:AP,Ya:Pr} requires quantitative characteristics of the architecture. Therefore, a specification language for the study of weighted architectures is needed and our wPCL is intended to play this role. wPCL consists of the PCL of \cite{Ma:Co} which is interpreted in the same way, and a copy of it which is interpreted quantitatively. This formulation has the advantage that practitioners can use the PCL exactly as they are used to, and the copy of it for the quantitative interpretation. The semantics of  wPCL formulas are polynomials with values in the semiring $K$. The semantics of the unweighted PCL formulas  take only the values $1$ and $0$ corresponding to $\mathrm{true}$ and $\mathrm{false}$, respectively. Weighted logics have been considered so far in other set-ups (cf. for instance \cite{Dr:Han}).    

The main contributions of our work are the following. 
\begin{itemize}
\item[i)] In our first main result, we prove that for every wPCL formula we can effectively construct, in doubly exponential time, an equivalent one in full normal form which is unique up to the equivalence relation. Computation of full normal forms can be done in an automatic way using the Maude system, following a method similar to the one for PCL in \cite{Ma:Co}.
\item[ii)] The second main result of our work states the decidability of  equivalence of wPCL formulas. Because of (i), two wPCL formulas are equivalent iff they have the same full normal form hence, we do not need to check equality for all possible models.      
\item[iii)] We introduce a notion of soundness (for the first time for weighted logics over semirings, according to our best knowledge), and as a third main result of the paper we show that wPCL is sound. 
\item[iv)] We succeed to describe in a strict logical way several well-known software architectures with additional quantitative characteristics, which are motivated by real applications. 
\end{itemize}
\noindent A preliminary version of this paper appeared in \cite{Pa:On}.

\section{Preliminaries}

A \emph{semiring} $(K,\oplus,\otimes,0,1)$\emph{ }consists of a set
$K,$ two binary operations $\oplus$ and $\otimes$ and two constant elements $0$
and $1$ such that $(K,\oplus,0)$ is a commutative monoid, $(K,\otimes,1)$ is a
monoid, multiplication $\otimes$ distributes over addition $\oplus$, and $0\otimes k=k\otimes 0=0$
for every $k\in K$. If the monoid $(K,\otimes,1)$ is commutative, then the
semiring is called commutative.\ The semiring is denoted simply by $K$ if the
operations and the constant elements are understood. The result of the empty
product as usual equals to $1$. The semiring $K$ is called (additively) idempotent if
$k\oplus k=k$ for every $k\in K$. The following algebraic structures are well-known semirings. The semiring $(\mathbb{N},+,\cdot,0,1)$ of natural numbers, the Boolean semiring $B=(\{0,1\},+,\cdot,0,1)$,  
the tropical or min-plus semiring $\mathbb{R}_{\min}=(\mathbb{R}_{+}\cup\{\infty\},\min,+,\infty,0)$ where $\mathbb{R}_{+}=\{r\in\mathbb{R}\mid r\geq0\}$, the arctical or max-plus semiring $\mathbb{R}_{\max}=(\mathbb{R}_{+}\cup\{-\infty\},\max,+,-\infty,0)$, the Viterbi semiring $(\left[  0,1\right]  ,\max,\cdot,0,1)$ used in probability theory, and every bounded distributive lattice with the operations $\sup$ and $\inf$.

Let $Q$ be a set. A \emph{formal series} (or simply \emph{series}) \emph{over}
$Q$ \emph{and} $K$ is a mapping $s:Q\rightarrow K$. The \emph{support of} $s$ is the set $\mathrm{supp}(s)=\{q \in Q \mid s(q) \neq 0  \}$. A series with finite support is also called a \emph{polynomial}. We denote by  $K\left\langle  Q
\right\rangle $ the class of all polynomials over $Q$ and $K$.  
Let $s,r\in K \left\langle Q \right\rangle $ and
$k\in K$. The \emph{sum} $s\oplus r$, the \emph{products with scalars} $ks$
and $sk$, and the\emph{ Hadamard product} $s\otimes r$\ are defined
elementwise, respectively  by $ s\oplus r(v)=s(v)\oplus r(v), ks(v)=k\otimes s(v), sk(v)=s(v)\otimes k, s\otimes r(v)=s(v)\otimes r(v)$  for every
$v\in Q$. \hide{It is a folklore result that the structure $\left(  K\left\langle
\left\langle Q\right\rangle \right\rangle ,\oplus,\otimes,\widetilde
{0},\widetilde{1}\right)$  is a semiring. Moreover, if $K$ is commutative (resp. idempotent),
then $\left(  K\left\langle \left\langle Q\right\rangle \right\rangle
,\oplus,\otimes,\widetilde{0},\widetilde{1}\right)  $ is also commutative (resp. idempotent).}
\begin{quote}
\emph{Throughout the paper }$(K, \oplus, \otimes, 0,1)$\emph{ will denote a commutative semiring.}
\end{quote}

\section{Weighted propositional interaction logic}

In this section, we introduce the weighted propositional interaction logic. For this, we
need to recall the propositional interaction logic \cite{Ma:Co} first.

Let $P$ be a nonempty finite set of \emph{ports}. We let $I(P)=\mathcal{P}%
(P)\setminus\{\emptyset\}$, where $\mathcal{P}(P)$ denotes the power set of $P$. Every set $a\in I(P)$ is called an
\emph{interaction}\label{interaction}. The syntax of \emph{propositional interaction logic}
(PIL for short) formulas over $P$\ is given by the grammar
\[
\phi::=\mathrm{true}\mid p\mid\overline{\phi}\mid\phi\vee\phi
\]
where $p\in P$. As usual, we set $\overline{\overline{\phi}}=\phi$ for every
PIL formula $\phi$ and $\mathrm{false}=\overline{\mathrm{true}}$. Then, the conjunction
of two PIL formulas $\phi,\phi^{\prime}$ is defined by $\phi\wedge
\phi^{\prime}=\overline{\left(  \overline{\phi}\vee\overline{\phi^{\prime}%
}\right)  }$. A PIL formula of the form $p_{1}\wedge\ldots\wedge p_{n}$
with $n>0$, and $p_{i}\in P$ or $p_i=\overline{p'_{i}}$ with $p'_i\in P$ for every $1\leq i\leq
n$, is called a \emph{monomial}. We shall simply denote a monomial
$p_{1}\wedge\ldots\wedge p_{n}$ by $p_{1}\ldots p_{n}$.

Let $\phi$ be a PIL formula and $a$ an interaction. We define the satisfaction relation 
$a\models_{i}\phi$ by induction on the structure of $\phi$ as follows:
\begin{itemize}
\item[-] $a\models_{i} \mathrm{true}$,

\item[-] $a\models_{i} p$  \ \ iff \ \ $p \in a$,

\item[-] $a\models_{i} \overline{\phi}$ \ \ iff \ \ $a \not \models_i \phi$,

\item[-] $a\models_{i} \phi_1 \vee \phi_2$ \ \ iff \ \ $a \models_i \phi_1$ or $a \models_i \phi_2$.
\end{itemize}

\noindent It should be clear that
$a\not \models _{i} \mathrm{false}$ for every $a\in I(P)$. For every interaction $a$ we
define its characteristic monomial $m_{a}%
=\bigwedge\nolimits_{p\in a}p\wedge\bigwedge\nolimits_{p\notin a}\overline{p}%
$. Then, for every $a^{\prime} \in I(P)$ we trivially get $a^{\prime
}\models_{i}m_{a}$ iff $a'= a$.

\begin{quote}
\emph{Throughout the paper }$P$\emph{ will denote a nonempty finite set of
ports.}
\end{quote}

\begin{definition}
\label{wPIL_syn}The syntax of formulas of the \emph{weighted PIL}
 (\emph{wPIL} for short) over $P$ and
$K$ is given by the grammar
\[
\varphi::=k\mid\phi\mid\varphi\oplus\varphi\mid\varphi\otimes\varphi
\]
where $k\in K$ and $\phi$ denotes a \emph{PIL} formula over $P$.
\end{definition}

We denote by $PIL(K,P)$ the set of all wPIL formulas $\varphi$
over $P$ and $K$. Next, we represent the semantics of formulas $\varphi\in
PIL(K,P)$ as polynomials $\left\Vert \varphi\right\Vert \in K
\left\langle I(P) \right\rangle $\footnote{Since $P$ is finite, the domain of $\|\varphi\|$ is finite and in turn its support is also finite.}. For the semantics of
PIL formulas $\phi$ we use the satisfaction relation as defined above.
In this way, we ensure that the semantics of PIL formulas $\phi$ gets
only the values $0$ and $1$.

\begin{definition}
\label{wPIL_sem}Let $\varphi\in PIL(K,P)$. The semantics of $\varphi$ is a
polynomial $\left\Vert \varphi\right\Vert \in K\left\langle
I(P)\right\rangle $ where for every $a\in I(P)$ the value $\left\Vert
\varphi\right\Vert (a)$ is defined inductively on the structure of $\varphi$ as follows:
\begin{itemize}
\item[-] $\left\Vert k\right\Vert (a)=k,$

\item[-] $\left\Vert \phi\right\Vert (a)=\left\{
\begin{array}
[c]{rl}%
1 & \textnormal{ if }a\models_{i}\phi\\
0 & \textnormal{ otherwise}%
\end{array}
,\right.  $

\item[-] $\left\Vert \varphi\oplus\psi\right\Vert (a)=\left\Vert
\varphi\right\Vert (a)\oplus\left\Vert \psi\right\Vert (a),$

\item[-] $\left\Vert \varphi\otimes\psi\right\Vert (a)=\left\Vert
\varphi\right\Vert (a)\otimes\left\Vert \psi\right\Vert (a).$
\end{itemize}
\end{definition}

The reader should note that the semantics of the wPIL formulas
$\phi\vee\phi$ and $\phi\oplus\phi$, where $\phi$ is a PIL formula, are
different. Indeed assume that $a\in I(P)$ is such that $a\models_{i}\phi$.
Then, by our definition above, we get $\left\Vert \phi\vee\phi\right\Vert
(a)=1$ whereas $\left\Vert \phi\oplus\phi\right\Vert (a)=1\oplus 1$.
Next we present an example of a wPIL formula.
\begin{example}
\label{m/s_ex}We recall from \cite{Ma:Co} the Master/Slave architecture
for two masters $M_{1},M_{2}$ and two slaves $S_{1},S_{2}$ with ports
$p_{1},p_{2}$ and $q_{1},q_{2}$, respectively (Figure \ref{MS_fig}). The monomial
$
\phi_{i,j}=m_{\{q_{i}, p_{j}\}} 
$
for every $1\leq i,j\leq 2$ defines the binary interaction between the ports $q_{i}$
and $p_{j}$.
For every $1\leq i,j\leq 2$ we consider a value $k_{i,j} \in K$ and the \emph{wPIL} formula $\varphi
_{i,j}=k_{i,j}\otimes\phi_{i,j}$. Hence, $k_{i,j}$\ can
be considered, according to the underlying semiring, as the ``cost" for the
implementation of the interaction $\phi_{i,j}$. For instance if $K$ is the
Viterbi semiring, then the value $k_{i,j}\in\lbrack0,1]$\ represents the
probability of the implementation of the interaction between the ports $q_{i}$
and $p_{j}$.

	\begin{figure}
		\begin{center}
				\begin{tikzpicture}
				\draw  (0,0) rectangle  (1,1);
				\draw  (1.7,0) rectangle  (2.7,1) ;
				\draw  (4,0) rectangle  (5,1) ;
				\draw  (5.7,0) rectangle  (6.7,1) ;
				\draw  (8,0) rectangle  (9,1) ;
				\draw  (9.7,0) rectangle  (10.7,1) ;
				\draw  (12,0) rectangle  (13,1) ;
				\draw  (13.7,0) rectangle  (14.7,1) ;
				
				\node at (0.5,0.7) {\small $M_1$};
				\node at (2.2,0.7) {\small $M_2$};
				\node at (4.5,0.7) {\small $M_1$};
				\node at (6.2,0.7) {\small $M_2$};
				\node at (8.5,0.7) {\small $M_1$};
				\node at (10.2,0.7) {\small $M_2$};
				\node at (12.5,0.7) {\small $M_1$};
				\node at (14.2,0.7) {\small $M_2$};
				
				\draw  (0.3,0) rectangle  (0.65,0.3);
				\draw  (2,0) rectangle  (2.35,0.3);
				\draw  (4.3,0) rectangle  (4.65,0.3);
				\draw  (6,0) rectangle  (6.36,0.3);
				\draw  (8.3,0) rectangle  (8.65,0.3);
				\draw  (10,0) rectangle  (10.35,0.3);
				\draw  (12.3,0) rectangle  (12.65,0.3);
				\draw  (14,0) rectangle  (14.35,0.3);
				
				\node at (0.5,0.15) {\tiny $p_1$};
				\node at (2.2,0.15) {\tiny $p_2$};
				\node at (4.5,0.15) {\tiny $p_1$};
				\node at (6.2,0.15) {\tiny $p_2$};
				\node at (8.5,0.15) {\tiny $p_1$};
				\node at (10.2,0.15) {\tiny $p_2$};
				\node at (12.5,0.15) {\tiny $p_1$};
				\node at (14.2,0.15) {\tiny $p_2$};

				\draw  (0,-3) rectangle  (1,-2);
				\draw  (1.7,-3) rectangle  (2.7,-2) ;
				\draw  (4,-3) rectangle  (5,-2) ;
				\draw  (5.7,-3) rectangle  (6.7,-2) ;
				\draw  (8,-3) rectangle  (9,-2) ;
				\draw  (9.7,-3) rectangle  (10.7,-2) ;
				\draw  (12,-3) rectangle  (13,-2) ;
				\draw  (13.7,-3) rectangle  (14.7,-2) ;
				
				\node at (0.5,-2.8) {\small $S_1$};
				\node at (2.2,-2.8) {\small $S_2$};
				\node at (4.5,-2.8) {\small $S_1$};
				\node at (6.2,-2.8) {\small $S_2$};
				\node at (8.5,-2.8) {\small $S_1$};
				\node at (10.2,-2.8) {\small $S_2$};
				\node at (12.5,-2.8) {\small $S_1$};
				\node at (14.2,-2.8) {\small $S_2$};
				
				\draw  (0.3,-2) rectangle  (0.65,-2.3);
				\draw  (2,-2) rectangle  (2.35,-2.3);
				\draw  (4.3,-2) rectangle  (4.65,-2.3);
				\draw  (6,-2) rectangle  (6.36,-2.3);
				\draw  (8.3,-2) rectangle  (8.65,-2.3);
				\draw  (10,-2) rectangle  (10.35,-2.3);
				\draw  (12.3,-2) rectangle  (12.65,-2.3);
				\draw  (14,-2) rectangle  (14.35,-2.3);
				
				\node at (0.5,-2.2) {\tiny $q_1$};
				\node at (2.2,-2.2) {\tiny $q_2$};
				\node at (4.5,-2.2) {\tiny $q_1$};
				\node at (6.2,-2.2) {\tiny $q_2$};
				\node at (8.5,-2.2) {\tiny $q_1$};
				\node at (10.2,-2.2) {\tiny $q_2$};
				\node at (12.5,-2.2) {\tiny $q_1$};
				\node at (14.2,-2.2) {\tiny $q_2$};
				
				\draw[fill] (0.5,0) circle [radius=1.5pt];
				\draw[fill] (2.2,0) circle [radius=1.5pt];
				\draw[fill] (4.5,0) circle [radius=1.5pt];
				\draw[fill] (6.2,0) circle [radius=1.5pt];
				\draw[fill] (8.5,0) circle [radius=1.5pt];
				\draw[fill] (10.2,0) circle [radius=1.5pt];
				\draw[fill] (12.5,0) circle [radius=1.5pt];
				\draw[fill] (14.2,0) circle [radius=1.5pt];

				\draw[fill] (0.5,-2) circle [radius=1.5pt];
				\draw[fill] (2.2,-2) circle [radius=1.5pt];
				\draw[fill] (4.5,-2) circle [radius=1.5pt];
				\draw[fill] (6.2,-2) circle [radius=1.5pt];
				\draw[fill] (8.5,-2) circle [radius=1.5pt];
				\draw[fill] (10.2,-2) circle [radius=1.5pt];
				\draw[fill] (12.5,-2) circle [radius=1.5pt];
				\draw[fill] (14.2,-2) circle [radius=1.5pt];
				
				\draw  (0.5,0)--(0.5,-2);
				\node at (0.1,-0.8) { $k_{1,1}$};
				\draw  (2.2,0)--(2.2,-2);
				\node at (1.8,-0.8) { $k_{2,2}$};
				
				\draw  (4.5,0)--(4.5,-2);
				\node at (4.1,-0.8) { $k_{1,1}$};
				\draw  (4.5,0)--(6.2,-2);
				\node at (5.6,-0.8) { $k_{2,1}$};

				\draw  (8.5,0)--(10.2,-2);
				\node at (8.7,-0.8) { $k_{2,1}$};
				\draw  (10.2,0)--(8.5,-2);
				\node at (10.1,-0.8) { $k_{1,2}$};
				
				\draw  (14.2,0)--(12.5,-2);
				\node at (13,-0.8) { $k_{1,2}$};
				\draw  (14.2,0)--(14.2,-2);
				\node at (14.6,-0.8) { $k_{2,2}$};
			\end{tikzpicture}
		\end{center}
		\caption{Weighted Master/Slave architecture.}
		\label{MS_fig}
	\end{figure}
\end{example}

\section{Weighted propositional configuration logic}

In this section, we introduce and investigate the weighted propositional configuration
logic. Firstly, we recall the configuration logic from \cite{Ma:Co}. More
precisely, the syntax of \emph{propositional configuration logic }(PCL for short) formulas over $P$\ is given by the grammar%
\[
f::=\mathrm{true}\mid\phi\mid\lnot f\mid f\sqcup f\mid f+f
\]
where $\phi$ denotes a PIL formula over $P$. The operators $\lnot$, $\sqcup$, and
$+$ are called \emph{complementation, union, }and \emph{coalescing},
respectively. We define also the \emph{intersection} $\sqcap$ and
\emph{implication} $\implies$ operators, respectively as follows:

- $f_{1}\sqcap f_{2}:=\lnot(\lnot f_{1}\sqcup\lnot f_{2}),$

- $f_{1}\implies f_{2}:=\lnot f_{1}\sqcup f_{2}.$

\noindent To avoid any confusion, every PCL formula which is a
PIL formula will be called an \emph{interaction formula.} We let
$C(P)=\mathcal{P}(I(P))\setminus\{\emptyset\}$. For every PCL formula
$f$ and $\gamma\in C(P)$ we define the satisfaction relation $\gamma\models
f$\ inductively on the structure of $f$ as follows:

\begin{itemize}
\item[-] $\gamma\models \mathrm{true},$

\item[-] $\gamma\models\phi$ \ \ iff \ $\ a\models_{i}\phi$ for every
$a\in\gamma,$

\item[-] $\gamma\models\lnot f$ \ \ iff \ \ $\gamma\not \models f,$

\item[-] $\gamma\models f_{1}\sqcup f_{2}$ \ \ iff \ \ $\gamma\models f_{1}$
or $\gamma\models f_{2},$

\item[-] $\gamma\models f_{1}+f_{2}$ \ \ iff \ \ there exist $\gamma
_{1},\gamma_{2}\in C(P)$ such that $\gamma=\gamma_{1}\cup\gamma_{2}$, and
$\gamma_{1}\models f_{1}$ and $\gamma_{2}\models f_{2}.$
\end{itemize}

\noindent Trivially, we get

- $\gamma\models f_{1}\sqcap f_{2}$ \ \ iff \ \ $\gamma\models f_{1}$
and $\gamma\models f_{2},$ and

- $\gamma\models f_{1}\implies f_{2}$ \ \ iff \ \ $\gamma
\not \models f_{1}$ or $\gamma\models f_{2}.$

\noindent We define the \emph{closure} ${\sim} f$ of every
PCL formula $f$ by

- ${\sim} f:=f+\mathrm{true}$,

\noindent and the \emph{disjunction} $f_{1}\vee f_{2}$ of two PCL
formulas $f_{1}$ and $f_{2}$ by

- $f_{1}\vee f_{2}:=f_{1}\sqcup f_{2}\sqcup (f_{1}+f_{2})$.

\noindent Two PCL formulas $f,f^{\prime}$ are called \emph{equivalent}, and we
denote it by $f\equiv f^{\prime},$\ whenever $\gamma\models f$\ iff
$\gamma\models f^{\prime}$\ for every $\gamma\in C(P)$. A PCL formula $f$ is called
\emph{downward-closed }if $\gamma\models f$ implies $\gamma^{\prime
}\models f$ for every $\gamma^{\prime}\subseteq\gamma$, and $f$ is called $\cup$\emph{-closed} if $\gamma\models f$ and $\gamma^{\prime}\models f$
implies $\gamma\cup\gamma^{\prime}\models f$.

In the subsequent Propositions \ref{char_inter}-\ref{col_prop}
we recall from \cite{Ma:Co}  properties of PCL  formulas that will be needed in the sequel. We refer the reader to \cite{Ma:Co} for further properties of PCL formulas. 

\begin{proposition}
\label{char_inter}A \emph{PCL} formula is $\cup$-closed and downward-closed iff it is
an interaction formula.
\end{proposition}

\begin{proposition}
\label{conj-inter}Let $\phi,\phi^{\prime}$ be interaction formulas. Then
\begin{itemize}
\item[i)] $\phi+\phi\equiv\phi$.
\item[ii)] $\phi\wedge\phi^{\prime}\equiv\phi\sqcap\phi^{\prime}$.
\end{itemize}
\end{proposition}

\noindent Due to Proposition \ref{conj-inter}(ii) we denote, in the sequel, both conjunction
and intersection operations of PCL formulas with the same symbol $\wedge$.

\begin{proposition}
\begin{itemize}
\item[(i)] \label{col_prop} The operators $\sqcup$, $\lnot$, $\wedge$ satisfy
the usual axioms of propositional logic.

\item[(ii)] The coalescing operation is associative, commutative, and has
$\mathrm{false}$ as an absorbing element.
\end{itemize}
\end{proposition}

Next we introduce our weighted PCL.

\begin{definition}
\label{wPCL_syn}The syntax of formulas of the \emph{weighted PCL} (\emph{wPCL} for short) over $P$ and
$K$ is given by the grammar%
\[
\zeta::=k\mid f\mid\zeta\oplus\zeta\mid\zeta\otimes\zeta\mid\zeta\uplus\zeta
\]
where $k\in K$, $f$ denotes a \emph{PCL} formula over $P$, and $\uplus$ denotes the
\emph{coalescing} operator among \emph{wPCL} formulas.
\end{definition}

Again, as for PCL formulas, to avoid any confusion, every wPCL formula which is a wPIL formula will be called a
\emph{weighted interaction formula.} We denote by $PCL(K,P)$ the set of all
wPCL formulas over $P$ and $K$. We represent the semantics of
formulas $\zeta\in PCL(K,P)$ as polynomials $\left\Vert \zeta\right\Vert \in
K\left\langle C(P)\right\rangle $. For the
semantics of PCL formulas we use the satisfaction relation as defined previously.

\begin{definition}
\label{wPCL_sem}Let $\zeta\in PCL(K,P)$. The semantics of $\zeta$\ is a polynomial
$\left\Vert \zeta\right\Vert \in K\left\langle C(P)
\right\rangle $ where for every $\gamma\in C(P)$ the value $\left\Vert
\zeta\right\Vert (\gamma)$ is defined inductively as follows:

\begin{itemize}
\item[-] $\left\Vert k\right\Vert (\gamma)=k,$

\item[-] $\left\Vert f\right\Vert (\gamma)=\left\{
\begin{array}
[c]{rl}%
1 & \textnormal{ if }\gamma\models f\\
0 & \textnormal{ otherwise}%
\end{array}
,\right.  $

\item[-] $\left\Vert \zeta_{1}\oplus\zeta_{2}\right\Vert (\gamma)=\left\Vert
\zeta_{1}\right\Vert (\gamma)\oplus\left\Vert \zeta_{2}\right\Vert (\gamma),$

\item[-] $\left\Vert \zeta_{1}\otimes\zeta_{2}\right\Vert (\gamma)=\left\Vert
\zeta_{1}\right\Vert (\gamma)\otimes\left\Vert \zeta_{2}\right\Vert (\gamma),$

\item[-] $\left\Vert \zeta_{1}\uplus\zeta_{2}\right\Vert (\gamma
)=\bigoplus\nolimits_{\gamma=\gamma_{1}\cupdot\gamma_{2}}\left(  \left\Vert
\zeta_{1}\right\Vert (\gamma_{1})\otimes\left\Vert \zeta_{2}\right\Vert
(\gamma_{2})\right)  $
\end{itemize}
\noindent where $\cupdot$ denotes the disjoint union of the sets $\gamma_1$ and $\gamma_2$.
\end{definition}

\noindent Since the semantics of every wPCL formula is defined on $C(P)$,
the sets $\gamma_{1}$ and $\gamma_{2}$ in $\left\Vert \zeta_{1}\uplus\zeta
_{2}\right\Vert (\gamma)$\ are nonempty. Two wPCL formulas $\zeta_{1},\zeta_{2}$ are
called equivalent, and we write $\zeta_{1}\equiv\zeta_{2}$, whenever
$\left\Vert \zeta_{1}\right\Vert =\left\Vert \zeta_{2}\right\Vert $.
The \emph{closure }${\sim}\zeta$ of every wPCL formula $\zeta\in
PCL(K,P)$ is determined by:

- ${\sim}\zeta:=\zeta \oplus (\zeta\uplus 1)$.

\begin{lemma}\label{clos_form} Let $\zeta \in PCL(K,P)$. Then 
$$\left\Vert {\sim}\zeta\right\Vert (\gamma)=\bigoplus
\nolimits_{\gamma^{\prime}\subseteq\gamma}\left\Vert \zeta\right\Vert
(\gamma^{\prime})$$
for every $\gamma \in C(P)$.
\end{lemma}
\begin{proof}
We compute
\begin{align*}
\left\Vert{\sim}\zeta\right\Vert (\gamma) &= \left\Vert \zeta \oplus (\zeta\uplus 1)\right\Vert (\gamma) 
 =\left\Vert \zeta  \right\Vert(\gamma) \oplus \left\Vert\zeta\uplus 1 \right\Vert (\gamma) \\
 & = \left\Vert \zeta  \right\Vert (\gamma)\oplus \left( \bigoplus\nolimits_{\gamma=\gamma'\cupdot\gamma''}\left(  \left\Vert
\zeta\right\Vert (\gamma')\otimes\left\Vert 1\right\Vert
(\gamma'')\right)  \right) \\
& = \left\Vert \zeta  \right\Vert (\gamma)\oplus  \bigoplus\nolimits_{\gamma' \varsubsetneq \gamma} \left\Vert
\zeta\right\Vert (\gamma') \\
&= \bigoplus\nolimits_{\gamma' \subseteq \gamma} \left\Vert
\zeta\right\Vert (\gamma')
\end{align*}
for every $\gamma \in C(P)$, where the fourth equality holds since $\gamma'$ and $\gamma''$ are disjoint.
\end{proof}

\begin{remark}
The reader should notice the difference in the definition of semantics among coalescing in \emph{PCL}, taken from \cite{Ma:Co}, and  coalescing in \emph{wPCL}. Though $(\gamma_1, \gamma_2)$ is not required to be a partition of $\gamma$ for coalescing in \emph{PCL}, it is used like that in all considered software architectures \cite{Ma:Co}. More precisely, let us assume that $f_1+f_2$ is a subformula of a \emph{PCL} formula $f$ over $P$ representing an architecture. This means that $f_1$ describes the topology of a part of the architecture and $f_2$ the topology of another part of the architecture. The two parts may have common components but not the same connections (interactions). Hence, if $\gamma \models f_1+f_2$,   then a subset $\gamma_1 $ of $ \gamma$ satisfies  $f_1$ and its complement $\gamma_2$ (according to $\gamma$) satisfies $f_2$.  Therefore, our definition of coalescing in \emph{wPCL} does not decrease its usefulness to the description of architectures. On the other hand, our definition intents also to the truth of Lemma \ref{clos_form} which is important for the description of architectures with quantitative characteristics, as it is the corresponding result in \emph{PCL} (cf. Statement after Definition 4.22 in \cite{Ma:Co}).    
\end{remark}

\hide{
For every PCL formula $f$\ over $P$ and every wPCL
formula $\zeta\in PCL(K,P)$, we consider the macro:

- $f\implies\zeta:=\lnot f\oplus(f\otimes\zeta).$

\noindent Then for $\gamma\in C(P)$, we get $\left\Vert f\implies
\zeta\right\Vert (\gamma)=\left\Vert \zeta\right\Vert (\gamma)$ if
$\gamma\models f$, and $\left\Vert f\implies\zeta\right\Vert (\gamma)=1$ otherwise.}

\begin{example}
[Example \ref{m/s_ex} continued]\label{ex_m_s_w}The four possible
configurations of the Master/Slave architecture for two masters
$M_{1},M_{2}$ and two slaves $S_{1},S_{2}$ with ports $p_{1},p_{2}$ and
$q_{1},q_{2}$, respectively (Figure \ref{MS_fig}), are given by the \emph{PCL} formula $
\left(  \phi_{1,1}\sqcup\phi_{1,2}\right)  +\left(  \phi_{2,1}\sqcup\phi
_{2,2}\right)  .$ 
We consider the \emph{wPCL} formula%
\[
\zeta={\sim}\left(  \left(  \varphi_{1,1}\oplus\varphi_{1,2}\right)
\uplus\left(  \varphi_{2,1}\oplus\varphi_{2,2}\right)  \right)  .
\]
Then for $\gamma\in C\left(  \left\{  p_{1},p_{2},q_{1},q_{2}\right\}
\right)  $ we get
\begin{align*}
\left\Vert \zeta\right\Vert (\gamma)  &  =\left\Vert {\sim}\left(  \left(
\varphi_{1,1}\oplus\varphi_{1,2}\right)  \uplus\left(  \varphi_{2,1}%
\oplus\varphi_{2,2}\right)  \right)  \right\Vert (\gamma)\\
&  =\bigoplus_{\gamma^{\prime}\subseteq\gamma}\left\Vert \left(  \varphi
_{1,1}\oplus\varphi_{1,2}\right)  \uplus\left(  \varphi_{2,1}\oplus
\varphi_{2,2}\right)  \right\Vert (\gamma^{\prime})\\
&  =\bigoplus_{\gamma^{\prime}\subseteq\gamma}\left(  \bigoplus_{\gamma
^{\prime}=\gamma_{1}\cupdot \gamma_{2}}\left(  \left\Vert \varphi_{1,1}%
\oplus\varphi_{1,2}\right\Vert (\gamma_{1})\otimes\left\Vert \varphi
_{2,1}\oplus\varphi_{2,2}\right\Vert (\gamma_{2})\right)  \right) \\
&  =\bigoplus_{\gamma^{\prime}\subseteq\gamma}\left(  \bigoplus_{\gamma
^{\prime}=\gamma_{1}\cupdot \gamma_{2}}\left(  \left(  \left\Vert \varphi
_{1,1}\right\Vert (\gamma_{1})\oplus\left\Vert \varphi_{1,2}\right\Vert
(\gamma_{1})\right)  \otimes\left(  \left\Vert \varphi_{2,1}\right\Vert
(\gamma_{2})\oplus\left\Vert \varphi_{2,2}\right\Vert (\gamma_{2})\right)
\right)  \right)  .
\end{align*}

\noindent Let for instance $K=
\mathbb{R}
_{\min}$ and $\gamma=\{\{q_{1},p_{1}\},\{q_{1},p_{2}\},\{q_{2}%
,p_{1}\},\{q_{2},p_{2}\}\}$. Then the value
\[
\left\Vert \zeta\right\Vert (\gamma)=\underset{\gamma^{\prime}\subseteq\gamma
}{\min}\left(  \underset{\gamma^{\prime}=\gamma_{1}\cupdot \gamma_{2}}{\min
}\left(  \min\left( \| \varphi_{1,1}\|(\gamma_{1}),\| \varphi_{1,2} \|(\gamma
_{1})\right)  +\min\left( \| \varphi_{2,1}\| (\gamma_{2}), \| \varphi_{2,2} \|(\gamma
_{2})\right)  \right)  \right)
\]
is the minimum ``cost" among all the implementations of the Master/Slave architecture.
\end{example}

In the sequel, we state several properties of our wPCL formulas.

\begin{proposition}
\label{wcoal_prop}Let $\zeta_{1},\zeta_{2},\zeta_{3}\in PCL(K,P)$. Then

\begin{itemize}
\item[(i)] $\left(  \zeta_{1}\uplus\zeta_{2}\right)  \uplus\zeta_{3}%
\equiv\zeta_{1}\uplus\left(  \zeta_{2}\uplus\zeta_{3}\right)  .$

\item[(ii)] $\zeta_{1}\uplus0\equiv0.$

\item[(iii)] $\zeta_{1}\uplus
\zeta_{2}\equiv\zeta_{2}\uplus\zeta_{1}.$

\item[(iv)] $\zeta_{1}\uplus(\zeta_{2}\oplus\zeta_{3})\equiv(\zeta_{1}\uplus\zeta
_{2})\oplus(\zeta_{1}\uplus\zeta_{3}).
$
\end{itemize}
\end{proposition}

\begin{proof}
We prove only (i) and (iv), the other two statements are straightforward. Let
$\gamma\in C(P)$. Then
\begin{align*}
\left\Vert \left(  \zeta_{1}\uplus\zeta_{2}\right)  \uplus\zeta_{3}\right\Vert
(\gamma)  &  =\underset{\gamma=\gamma^{\prime}\cupdot\gamma_{3}}{\bigoplus
}\left(  \left\Vert \zeta_{1}\uplus\zeta_{2}\right\Vert (\gamma^{\prime
})\otimes\left\Vert \zeta_{3}\right\Vert (\gamma_{3})\right) \\
&  =\underset{\gamma=\gamma^{\prime}\cupdot\gamma_{3}}{\bigoplus}\left(  \left(
\underset{\gamma^{\prime}=\gamma_{1}\cupdot\gamma_{2}}{\bigoplus}\left(
\left\Vert \zeta_{1}\right\Vert (\gamma_{1})\otimes\left\Vert \zeta
_{2}\right\Vert (\gamma_{2})\right)  \right)  \otimes\left\Vert \zeta
_{3}\right\Vert (\gamma_{3})\right) \\
&  =\underset{\gamma=\gamma_{1}\cupdot\gamma_{2}\cupdot\gamma_{3}}{\bigoplus}\left(
\left(  \left\Vert \zeta_{1}\right\Vert (\gamma_{1})\otimes\left\Vert
\zeta_{2}\right\Vert (\gamma_{2})\right)  \otimes\left\Vert \zeta
_{3}\right\Vert (\gamma_{3})\right) \\
&  =\underset{\gamma=\gamma_{1}\cupdot\gamma_{2}\cupdot\gamma_{3}}{\bigoplus}\left(
\left\Vert \zeta_{1}\right\Vert (\gamma_{1})\otimes\left(  \left\Vert
\zeta_{2}\right\Vert (\gamma_{2})\otimes\left\Vert \zeta_{3}\right\Vert
(\gamma_{3})\right)  \right) \\
&  =\underset{\gamma=\gamma_{1}\cupdot\gamma\prime}{\bigoplus}\left(  \left\Vert
\zeta_{1}\right\Vert (\gamma_{1})\otimes \left(\underset{\gamma^{\prime}=\gamma
_{2}\cupdot\gamma_{3}}{\bigoplus}\left(  \left\Vert \zeta_{2}\right\Vert
(\gamma_{2})\otimes\left\Vert \zeta_{3}\right\Vert (\gamma_{3})\right)\right)
\right) \\
&  =\left\Vert \zeta_{1}\uplus\left(  \zeta_{2}\uplus\zeta_{3}\right)
\right\Vert (\gamma),
\end{align*}
and hence $\left(  \zeta_{1}\uplus\zeta_{2}\right)  \uplus\zeta_{3}\equiv
\zeta_{1}\uplus\left(  \zeta_{2}\uplus\zeta_{3}\right)  $,\ as required. 

Next 
\begin{align*}
\left\Vert \zeta_{1}\uplus(\zeta_{2}\oplus\zeta_{3})\right\Vert (\gamma)  &
=\underset{\gamma=\gamma_{1}\cupdot\gamma^{\prime}}{\bigoplus}\left(  \left\Vert
\zeta_{1}\right\Vert (\gamma_{1})\otimes\left\Vert \zeta_{2}\oplus\zeta
_{3}\right\Vert (\gamma^{\prime})\right) \\
&  =\underset{\gamma=\gamma_{1}\cupdot\gamma^{\prime}}{\bigoplus}\left(
\left\Vert \zeta_{1}\right\Vert (\gamma_{1})\otimes\left(  \left\Vert
\zeta_{2}\right\Vert (\gamma^{\prime})\oplus\left\Vert \zeta_{3}\right\Vert
(\gamma^{\prime})\right)  \right) \\
&  =\underset{\gamma=\gamma_{1}\cupdot\gamma^{\prime}}{\bigoplus}\left(  \left(
\left\Vert \zeta_{1}\right\Vert (\gamma_{1})\otimes\left\Vert \zeta
_{2}\right\Vert (\gamma^{\prime})\right)  \oplus\left(  \left\Vert \zeta
_{1}\right\Vert (\gamma_{1})\otimes\left\Vert \zeta_{3}\right\Vert
(\gamma^{\prime})\right)  \right) \\
&  =\underset{\gamma=\gamma_{1}\cupdot\gamma^{\prime}}{\bigoplus}\left(
\left\Vert \zeta_{1}\right\Vert (\gamma_{1})\otimes\left\Vert \zeta
_{2}\right\Vert (\gamma^{\prime})\right)  \oplus\underset{\gamma=\gamma
_{1}\cupdot\gamma^{\prime}}{\bigoplus}\left(  \left\Vert \zeta_{1}\right\Vert
(\gamma_{1})\otimes\left\Vert \zeta_{3}\right\Vert (\gamma^{\prime})\right) \\
&  =\left\Vert \zeta_{1}\uplus\zeta_{2}\right\Vert (\gamma)\oplus\left\Vert
\zeta_{1}\uplus\zeta_{3}\right\Vert (\gamma)\\
&  =\left\Vert \left(  \zeta_{1}\uplus\zeta_{2}\right)  \oplus\left(
\zeta_{1}\uplus\zeta_{3}\right)  \right\Vert (\gamma),
\end{align*}
and we are done.
\end{proof}

We aim to show that $\otimes$ does not distribute, in general, over $\uplus$.
For this, we consider the semiring $\left(  \mathbb{N},+,\cdot,0,1\right)  $
of natural numbers, the set of ports $P=\{p,q\}$ and the formulas $\zeta,\zeta_{1},\zeta_{2}\ \in PCL(\mathbb{N},P)$ with $\zeta=2$ and $\zeta_1=\zeta_2=1$. We let $\gamma=\{\{p\},\{q\}\}$. Then  
\begin{align*}
\Vert\zeta\otimes(\zeta_{1}\uplus\zeta_{2})\Vert(\gamma)    =\Vert\zeta
\Vert(\gamma)\cdot\Vert\zeta_{1}\uplus\zeta_{2}\Vert(\gamma)
  =\Vert\zeta\Vert(\gamma)\cdot\left(  \underset{\gamma=\gamma_{1}\cupdot
\gamma_{2}}{\sum}\left(  \Vert\zeta_{1}\Vert(\gamma_{1})\cdot\Vert\zeta
_{2}\Vert(\gamma_{2})\right)  \right) 
  =2 \cdot(1+1)=4
\end{align*}
whereas 
\begin{align*}
\Vert(\zeta\otimes\zeta_{1})\uplus(\zeta\otimes\zeta_{2})\Vert(\gamma)  =\underset{\gamma=\gamma_{1}\cupdot \gamma_{2}}{\sum}\left(  \Vert\zeta
\otimes\zeta_{1}\Vert(\gamma_{1})\cdot\Vert\zeta\otimes\zeta_{2}\Vert
(\gamma_{2})\right) 
  =2\cdot2 + 2\cdot2=8.
\end{align*}
Hence $\zeta\otimes(\zeta_{1}\uplus\zeta_{2})
\not \equiv (\zeta\otimes\zeta_{1})\uplus(\zeta\otimes\zeta_{2})$.
Nevertheless, this is not the case whenever $\zeta$ is a PIL formula.
More precisely, we state the subsequent proposition.

\begin{proposition}
\label{conj_over_coal} Let $\phi$ be a \emph{PIL} formula over $P$ and $\zeta_{1},\zeta
_{2}\ \in PCL(K,P)$. Then
\[
\phi\otimes(\zeta_{1}\uplus\zeta_{2})\equiv(\phi\otimes\zeta_{1})\uplus
(\phi\otimes\zeta_{2}).
\]

\end{proposition}

\begin{proof}
For every $\gamma\in C(P)$ we have%
\begin{align*}
\left\Vert \phi\otimes(\zeta_{1}\uplus\zeta_{2})\right\Vert (\gamma)  &
=\left\Vert \phi\right\Vert (\gamma)\otimes\left\Vert \zeta_{1}\uplus\zeta
_{2}\right\Vert (\gamma)\\
&  =\left\Vert \phi\right\Vert (\gamma)\otimes\left(  \underset{\gamma
=\gamma_{1}\cupdot \gamma_{2}}{\bigoplus}\left(  \left\Vert \zeta_{1}\right\Vert
(\gamma_{1})\otimes\left\Vert \zeta_{2}\right\Vert (\gamma_{2})\right)
\right) \\
&  =\underset{\gamma=\gamma_{1}\cupdot \gamma_{2}}{\bigoplus}\left(  \left\Vert
\phi\right\Vert (\gamma)\otimes(\left\Vert \zeta_{1}\right\Vert (\gamma
_{1})\otimes\left\Vert \zeta_{2}\right\Vert (\gamma_{2}))\right)  .
\end{align*}
We distinguish two cases.

\begin{itemize}
\item $\Vert\phi\Vert(\gamma)=1$. Then by definition, $\gamma\models\phi$
which, by Proposition \ref{char_inter}, implies that $\gamma^{\prime}\models\phi$, and hence $\Vert\phi
\Vert(\gamma^{\prime})=1$\ for every $\gamma^{\prime}\subseteq\gamma$.
Therefore, we get
\begin{align*}
&  \underset{\gamma=\gamma_{1}\cupdot \gamma_{2}}{\bigoplus}\left(  \left\Vert
\phi\right\Vert (\gamma)\otimes(\left\Vert \zeta_{1}\right\Vert (\gamma
_{1})\otimes\left\Vert \zeta_{2}\right\Vert (\gamma_{2}))\right)  \text{ }\\
&  =\underset{\gamma=\gamma_{1}\cupdot \gamma_{2}}{\bigoplus}\left(  \left\Vert
\phi\right\Vert (\gamma_{1})\otimes\left\Vert \zeta_{1}\right\Vert (\gamma
_{1})\otimes\Vert\phi\Vert(\gamma_{2})\otimes\left\Vert \zeta_{2}\right\Vert
(\gamma_{2})\right) \\
&  =\underset{\gamma=\gamma_{1}\cupdot \gamma_{2}}{\bigoplus}\left(  \Vert
\phi\otimes\zeta_{1}\Vert(\gamma_{1})\otimes\Vert\phi\otimes\zeta_{2}%
\Vert(\gamma_{2})\right) \\
&  =\Vert(\phi\otimes\zeta_{1})\uplus(\phi\otimes\zeta_{2})\Vert(\gamma).
\end{align*}

\item $\Vert\phi\Vert(\gamma)=0$. Hence $\gamma\not \models \phi$, i.e., there
is an $a\in\gamma$ such that $a\not \models _{i}\phi$. This in turn implies
that $\gamma^{\prime}\not \models \phi$ for every $\gamma^{\prime}%
\subseteq\gamma$ with $a\in \gamma^{\prime}$. Therefore, we get
\[
\underset{\gamma=\gamma_{1}\cupdot \gamma_{2}}{\bigoplus}\left(  \left\Vert
\phi\right\Vert (\gamma)\otimes(\left\Vert \zeta_{1}\right\Vert (\gamma
_{1})\otimes\left\Vert \zeta_{2}\right\Vert (\gamma_{2}))\right)  =0
\]
and%
\[
\underset{\gamma=\gamma_{1}\cupdot \gamma_{2}}{\bigoplus}\left(  \left\Vert
\phi\right\Vert (\gamma_{1})\otimes\left\Vert \zeta_{1}\right\Vert (\gamma
_{1})\otimes\Vert\phi\Vert(\gamma_{2})\otimes\left\Vert \zeta_{2}\right\Vert
(\gamma_{2})\right)  =0,
\]
i.e.,
\[
\left\Vert \phi\otimes(\zeta_{1}\uplus\zeta_{2})\right\Vert (\gamma
)=0=\Vert(\phi\otimes\zeta_{1})\uplus(\phi\otimes\zeta_{2})\Vert(\gamma)
\]
and this concludes our proof. 
\end{itemize}
\end{proof}

\begin{proposition}
\label{wclos_prop} Let $\zeta,\zeta_{1},\zeta_{2}\in PCL(K,P)$. Then

\begin{itemize}
\item[(i)] ${\sim}(\zeta_{1}\oplus\zeta_{2})\equiv{\sim}\zeta_{1}\oplus{\sim}
\zeta_{2}.$

\end{itemize}

\noindent If in addition $K$ is idempotent, then

\begin{itemize}
\item[(ii)] ${\sim}(\zeta_{1}\uplus\zeta_{2})\equiv{\sim}\zeta_{1}\uplus{\sim}
\zeta_{2}.$

\item[(iii)] ${\sim}{\sim}\zeta\equiv{\sim}\zeta.$
\end{itemize}
\end{proposition}

\begin{proof}
For every $\gamma\in C(P)$ we have

(i)%
\begin{align*}
\left\Vert {\sim}(\zeta_{1}\oplus\zeta_{2})\right\Vert (\gamma)  &
=\underset{\gamma^{\prime}\subseteq\gamma}{\bigoplus}\left\Vert \zeta
_{1}\oplus\zeta_{2}\right\Vert (\gamma^{\prime})  =\underset{\gamma^{\prime}\subseteq\gamma}{\bigoplus}\left(  \left\Vert
\zeta_{1}\right\Vert (\gamma^{\prime})\oplus\left\Vert \zeta_{2}\right\Vert
(\gamma^{\prime}\right)  )\\
&  =\underset{\gamma^{\prime}\subseteq\gamma}{\bigoplus}\left\Vert \zeta
_{1}\right\Vert (\gamma^{\prime})\oplus\underset{\gamma^{\prime}%
\subseteq\gamma}{\bigoplus}\left\Vert \zeta_{2}\right\Vert (\gamma^{\prime}) =\left\Vert {\sim}\zeta_{1}\right\Vert (\gamma)\oplus\left\Vert {\sim}\zeta
_{2}\right\Vert (\gamma)\\
&  =\left\Vert {\sim}\zeta_{1}\oplus{\sim}\zeta_{2}\right\Vert (\gamma).
\end{align*}

(ii)%
\begin{align*}
\left\Vert {\sim}(\zeta_{1}\uplus\zeta_{2})\right\Vert (\gamma)  &
=\underset{\gamma^{\prime}\subseteq\gamma}{\bigoplus}\left\Vert \zeta
_{1}\uplus\zeta_{2}\right\Vert (\gamma^{\prime})\\
&  =\underset{\gamma^{\prime}\subseteq\gamma}{\bigoplus}\left(  \underset
{\gamma^{\prime}=\gamma_{1}^{\prime}\cupdot \gamma_{2}^{\prime}}{\bigoplus}\left(
\left\Vert \zeta_{1}\right\Vert (\gamma_{1}^{\prime})\otimes\left\Vert
\zeta_{2}\right\Vert (\gamma_{2}^{\prime})\right)  \right) \\
&  =\underset{\gamma=\gamma_{1}\cupdot \gamma_{2}}{\bigoplus}\left(
\underset{\gamma_{2}^{\prime}\subseteq\gamma_{2}}{\underset{\gamma_{1}%
^{\prime}\subseteq\gamma_{1}}{\bigoplus}}\left(  \left\Vert \zeta
_{1}\right\Vert (\gamma_{1}^{\prime})\otimes\left\Vert \zeta_{2}\right\Vert
(\gamma_{2}^{\prime})\right)  \right) \\
&  =\underset{\gamma=\gamma_{1}\cupdot \gamma_{2}}{\bigoplus}\left(  \left(
\underset{\gamma_{1}^{\prime}\subseteq\gamma_{1}}{\bigoplus}\Vert\zeta
_{1}\Vert(\gamma_{1}^{\prime})\right)  \otimes\left(  \underset{\gamma
_{2}^{\prime}\subseteq\gamma_{2}}{\bigoplus}\Vert\zeta_{2}\Vert(\gamma
_{2}^{\prime})\right)  \right) \\
&  =\underset{\gamma=\gamma_{1}\cupdot \gamma_{2}}{\bigoplus}\left(  \Vert
{\sim}\zeta_{1}\Vert(\gamma_{1})\otimes\Vert{\sim}\zeta_{2}\Vert(\gamma
_{2})\right) \\
&  =\left\Vert {\sim}\zeta_{1}\uplus{\sim}\zeta_{2}\right\Vert (\gamma).
\end{align*}
where the third equality holds since $K$ is idempotent\footnote{For instance, if $\gamma = \{a,b,c \}$, then the product $\|\zeta_1\|(a) \otimes \|\zeta_2\|(b)$ occurs only once in the left-hand side of the third equality, but two times in the right-hand side one.}.

(iii)
\begin{align*}
\left\Vert {\sim}{\sim}\zeta\right\Vert (\gamma)  &  =\underset{\gamma
_{1}\subseteq\gamma}{\bigoplus}\left\Vert {\sim}\zeta\right\Vert (\gamma_{1})  =\underset{\gamma_{1}\subseteq\gamma}{\bigoplus}\left(  \underset
{\gamma_{2}\subseteq\gamma_{1}}{\bigoplus}\left\Vert \zeta\right\Vert
(\gamma_{2})\right)   =\underset{\gamma_{1}\subseteq\gamma}{\bigoplus}\left\Vert \zeta\right\Vert
(\gamma_{1})  =\left\Vert {\sim}\zeta\right\Vert (\gamma)
\end{align*}
where the third equality holds by the idempotency property of $K$.
\end{proof}

In the subsequent examples, we present wPCL formulas describing well-known architectures equipped with quantitative features. 
\begin{example}
[Weighted Star architecture]		
		Star architecture is a software architecture relating components of the same type. More precisely, given a set of components one of them is considered as the  \emph{central component} and is connected to every other component through a binary interaction. No other interactions are permitted.		
		\begin{figure}
				\begin{center}
					\begin{tikzpicture}
					\draw  (0.7,-9.75) rectangle  (2,-9); 
					\draw  (7,-7.85) rectangle  (8.3,-7.1); 
					\draw  (7,-9) rectangle  (8.3,-8.25);
					\draw  (7,-10.15) rectangle  (8.3,-9.4);
					\draw  (7,-11.3) rectangle  (8.3,-10.55);

					\draw  (1.45,-9.55) rectangle  (2,-9.15); 
					\node at (1.7,-9.4) {$ s_1$};
					\draw[fill] (2,-9.35) circle [radius=2pt];
					
					\draw  (7,-7.7) rectangle  (7.57,-7.3); 
					\node at (7.35,-7.54) {$ s_2$};
					\draw[fill] (7.02,-7.5) circle [radius=2pt];
					
					\draw  (7,-8.85) rectangle  (7.57,-8.45); 
					\node at (7.35,-8.7) {$ s_3$};
					\draw[fill] (7.02,-8.65) circle [radius=2pt];
					
					\draw  (7,-10) rectangle  (7.57,-9.6); 
					\node at (7.35,-9.85) {$ s_4$};
					\draw[fill] (7.02,-9.8) circle [radius=2pt];
					
					\draw  (7,-11.15) rectangle  (7.57,-10.75); 
					\node at (7.35,-11) {$ s_5$};
					\draw[fill] (7.02,-10.95) circle [radius=2pt];
					
					\draw (2,-9.35)--(7.02,-7.5);
					\draw (2,-9.35)--(7.02,-8.65);
					\draw (2,-9.35)--(7.02,-9.8);
					\draw (2,-9.35)--(7.02,-10.95);
					
					\node at (5,-8.0) {$ k_2$};
					\node at (5,-8.7) {$ k_3$};
					\node at (5.05,-9.35) {$ k_4$};
					\node at (5.1,-10.1) {$ k_5$};

					\end{tikzpicture}
				\end{center}
			\caption{Weighted Star architecture.}
			\end{figure}
				Next, we build a \emph{wPCL} formula representing the Star architecture with five components. We assume that every component has a single port and let $I=\{1,2,3,4,5\}$ and $ P=\{ s_1, s_2, s_3, s_4, s_5 \} $. For simplicity we call every component by its port name. Let us consider firstly that $s_1$ is the central component. The architecture is illustrated in Figure 2. We denote by $k_i \in K$ the weight of the binary interaction between $s_1$ and $s_i$ for every $i \in I\setminus\{1\}$. The \emph{wPIL} formula characterizing this interaction, for every $i \in I\setminus\{1\}$, is given by $ \zeta_{i}= k_{i} \otimes m_{\{ s_1,  s_i \}}$. 
		Since $s_1$ interacts with all the other ports, the Star architecture is described by the subsequent \emph{wPCL} formula:
		\[ \zeta = \zeta_2 \uplus \zeta_3 \uplus \zeta_4 \uplus \zeta_5 . \]
		
\noindent Now, we consider the case where every component in the architecture can be the central one, hence we get five versions of the Star architecture with five components. By assuming that the weight of every interaction refers to the cost of its implementation, we should like to compute which version of the Star architecture has the minimum cost and which one has the  maximum cost. Therefore, we consider for every $i \in I$ the corresponding architecture with $s_i$ being the central component. We denote by $k_{ij} \in K$ the weight of the binary interaction between the central component $s_i$ and the component $s_j$ with $j\in I \setminus \{i\}$. The \emph{wPCL} formula characterizing this interaction is given by $ \zeta_{ij}= k_{ij} \otimes m_{\{s_i, s_j \}}$. Therefore, the \emph{wPCL} formula 
		\[ \zeta_i^\prime = \biguplus_{j\in I \setminus \{i\}} \zeta_{ij} \]
				describes the binary interactions of the central component $s_i$ with the rest of all the other components. We conclude to the next \emph{wPCL} formula $\zeta'$ which describes the five alternative versions of the Star architecture, derived by five components, and the total cost: 
		\[ \zeta^\prime = \sim \left(\zeta_1^\prime \oplus \zeta_2^\prime \oplus \zeta_3^\prime\oplus \zeta_4^\prime \oplus \zeta_5^\prime \right). \]
		Hence, for $K=\mathbb{R}_{\min}$ (resp. $K=\mathbb{R}_{\max}$) and $\gamma=\{\{s_i,s_j\} \mid i,j \in I \text{ and } i\neq j \}$ we get the architecture with the minimum (resp. maximum) cost.
				\end{example}

	\begin{example}[Weighted Pipes and Filters architecture]\label{P/F} 	The Pipes and Filters architecture (cf. \cite{Ga:An}) involves two types of components, the pipes and the filters denoted respectively, by the letters $P$ and $F$. Every $P$ component has two ports $p.in$ and $p.out$, and every $F$ component has also two ports $f.in$ and $f.out$ (Figure \ref{pi_fi_ex}). Every $f.in$ (resp. $f.out$) port of a filter is connected to a $p.out$ (resp. $p.in$) port of a single pipe. The $p.out$ port of any pipe can be connected to at most one filter port $f.in$. 
	
	\begin{figure}
			\begin{center}
				\begin{tikzpicture}
				\draw  (0,0) rectangle  (1.9,0.6);
				\draw  (3,0) rectangle  (4.9,0.6);
				\draw  (0,-1.3) rectangle  (1.9,-0.7);
				\draw  (3,-1.3) rectangle  (4.9,-0.7);
				\draw  (6,-0.65) rectangle  (7.9,-0.05);
				\draw  (9,-0.65) rectangle  (10.9,-0.05);
				\draw  (12,-0.65) rectangle  (13.9,-0.05);
				
				\node at (0.9,0.3) {\scriptsize $P$};
				\node at (3.9,0.3) {\scriptsize $F$};
				\node at (0.9,-1) {\scriptsize $P$};
				\node at (3.9,-1) {\scriptsize $F$};
				\node at (6.9,-0.35) {\scriptsize $P$};
				\node at (9.9,-0.35) {\scriptsize $F$};
				\node at (12.9,-0.35) {\scriptsize $P$};
				
				\draw  (0,0.15) rectangle  (0.55,0.45);
				\draw  (3,0.15) rectangle  (3.55,0.45);
				\draw  (0,-1.15) rectangle  (0.55,-0.85);
				\draw  (3,-1.15) rectangle  (3.55,-0.85);
				\draw  (6,-0.5) rectangle  (6.55,-0.2);
				\draw  (9,-0.5) rectangle  (9.55,-0.2);
				\draw  (12,-0.5) rectangle  (12.55,-0.2);
				
				\draw  (1.2,0.15) rectangle  (1.9,0.45);
				\draw  (4.2,0.15) rectangle  (4.9,0.45);
				\draw  (1.2,-1.15) rectangle  (1.9,-0.85);
				\draw  (4.2,-1.15) rectangle  (4.9,-0.85);
				\draw  (7.2,-0.5) rectangle  (7.9,-0.2);
				\draw  (10.2,-0.5) rectangle  (10.9,-0.2);
				\draw  (13.2,-0.5) rectangle  (13.9,-0.2);
				
				\node at (0.3,0.3) {\tiny $p.in$};
				\node at (1.56,0.3) {\tiny $p.out$};
				\node at (3.3,0.3) {\tiny $f.in$};
				\node at (4.55,0.3) {\tiny $f.out$};
				\node at (0.3,-1) {\tiny $p.in$};
				\node at (1.56,-1) {\tiny $p.out$};
				\node at (3.3,-1) {\tiny $f.in$};
				\node at (4.55,-1) {\tiny $f.out$};
				
				\node at (6.3,-0.35) {\tiny $p.in$};
				\node at (7.56,-0.35) {\tiny $p.out$};
				
				\node at (9.3,-0.35) {\tiny $f.in$};
				\node at (10.55,-0.35) {\tiny $f.out$};
				
				\node at (12.3,-0.35) {\tiny $p.in$};
				\node at (13.56,-0.35) {\tiny $p.out$};
				
				\draw[fill] (0,0.3) circle [radius=1pt];
				\draw[fill] (1.9,0.3) circle [radius=1pt];
				
				\draw[fill] (3,0.3) circle [radius=1pt];
				\draw[fill] (4.9,0.3) circle [radius=1pt];
				
				\draw[fill] (0,-1) circle [radius=1pt];
				\draw[fill] (1.9,-1) circle [radius=1pt];
				
				\draw[fill] (3,-1) circle [radius=1pt];
				\draw[fill] (4.9,-1) circle [radius=1pt];
				
				\draw[fill] (6,-0.35) circle [radius=1pt];
				\draw[fill] (7.9,-0.35) circle [radius=1pt];
				
				\draw[fill] (9,-0.35) circle [radius=1pt];
				\draw[fill] (10.9,-0.35) circle [radius=1pt];
				
				\draw[fill] (12,-0.35) circle [radius=1pt];
				\draw[fill] (13.9,-0.35) circle [radius=1pt];

				\draw (1.9,0.3)--(3,0.3);
				\draw (1.9,-1)--(3,-1);
				\draw (4.9,0.3)--(6,-0.35);
				\draw (4.9,-1)--(6,-0.35);
				\draw (7.9,-0.35)--(9,-0.35);
				\draw (10.9,-0.35)--(12,-0.35);
				\end{tikzpicture}
			\end{center}
			\caption{Pipes and Filters architecture.}
			\label{pi_fi_ex}
		\end{figure}
		
	In our example, we develop a \emph{wPCL} formula characterizing the Pipes and Filters architecture for four pipes and three filters. For this, we let $I=\{1,2,3\}$ and denote the three filter components by $F_1, F_2, F_3$ with ports $f_i.in$ and $f_i.out$ for $i\in I$, respectively. Similarly, we let $J=\{1,2,3,4\}$ and denote the four pipe components by $P_1, P_2, P_3, P_4$ with ports $p_j.in$ and $p_j.out$ for $j\in J$, respectively. Hence, the set $P$ of all ports is determined by $P=\{ f_i.in, f_i.out, p_j.in, p_j.out, \mid i \in I, \ j\in J \}$. For every $i\in I$ and $j\in J$ we shall denote by $k_{ij} \in K$ the weight of the interaction among the ports $f_i.in$ and $p_j.out$. Similarly, for every $i\in I$ and $j\in J$ the weight of the interaction between the ports $f_i.out$ and $p_j.in$ is denoted by $k_{ij}^\prime \in K.$ The \emph{wPCL} formula that formalizes the interactions between input ports of filters and output ports of pipes with the corresponding weights is
$$\zeta_1= \bigotimes_{i \in I} \left(\bigoplus_{j \in J} \left((\sim(f_i.in \wedge p_j.out)) \otimes k_{ij} \otimes \left (  \bigwedge_{j' \in J\setminus \{j\}} \overline{f_i.in \wedge p_{j'}.out} \right)  \right)\right)$$	
whereas the  \emph{wPCL} formula 
$$\zeta_2= \bigotimes_{i \in I}\left(\bigoplus_{j \in J} \left((\sim(f_i.out \wedge p_j.in)) \otimes k'_{ij} \otimes \left (  \bigwedge_{j' \in J\setminus \{j\}} \overline{f_i.out \wedge p_{j'}.in} \right)  \right) \right)$$	
describes the interactions and their weights among output ports of filters and input ports of pipes. Furthermore, we need to ensure that the output of a pipe can be connected to at most one input of a filter, and that the pipes can be connected only with filters and vice-versa. These requirements are satisfied respectively, by the following \emph{PCL} formulas:
$$g=\bigwedge_{j \in J} \left(\bigsqcup_{i \in I} \left(\bigwedge_{i' \in I \setminus \{i\}} \overline{p_j.out \wedge f_{i'}.in} \right)  \right)  \quad \text{ and } \quad g'=\bigvee_{i\in I, j\in J} \left ( m_{\{f_i.in,p_j.out\}} \vee  m_{\{f_i.out,p_j.in\}} \right).$$	
We conclude to the \emph{wPCL} formula 
	$$ \zeta = \sim \left( \zeta_1 \otimes \zeta_2 \otimes g  \right) \otimes g'$$
for the weighted Pipes and Filters architecture for four pipes and three filters. 
\hide{	
	Next let $K={\mathbb{R}}_{\min}$ and $\gamma \in C(P)$ be a set of interactions.  Then we have 
	\begin{align*}
	\left\Vert \zeta\right\Vert (\gamma) & = \left\Vert \sim \left( \zeta_1 \otimes \zeta_2 \otimes g  \right) \otimes g' \right\Vert (\gamma) \\ 
	& = \left\Vert \sim \left( \zeta_1 \otimes \zeta_2 \otimes g  \right) \right\Vert (\gamma) + \left\Vert g' \right\Vert (\gamma)  \\ 
	& = \min_{\gamma^\prime \subseteq \gamma}\left( \left\Vert \zeta_1 \otimes \zeta_2 \otimes g  \right\Vert (\gamma^\prime)  \right)  + \left\Vert g' \right\Vert (\gamma) \\ & = \min_{\gamma^\prime \subseteq \gamma}\left( \left\Vert \zeta_1\right\Vert (\gamma^\prime)  + \left\Vert  \zeta_2 \right\Vert (\gamma^\prime) + \left\Vert  g  \right\Vert (\gamma^\prime) \right )  + \left\Vert g' \right\Vert (\gamma).
	\end{align*}
	For instance, if 
	\begin{align*}
	\gamma =  & \{ \{f_1.in , p_3.out\}, \{f_3.in , p_1.out\}, \{f_2.in , p_3.out\}, \{f_1.in , p_2.out\}, \{f_2.in , p_2.out\},  \\ &    \{ f_1.out, p_1.in\},   \{ f_2.out, p_1.in\}, \{ f_2.out, p_3.in\}, \{ f_3.out, p_2.in\}, \{ f_3.out, p_3.in\}   \},
	\end{align*}
	then the value $
	\left\Vert \zeta \right\Vert (\gamma)$
	returns the Pipes and Filters architecture with the minimum weight among all the ones that can be constructed with the given set of interactions $\gamma$.}
	\end{example}

\begin{example}
	[Weighted Publish/Subscribe architecture]\label{P/S} \emph{Publish/Subscribe} is a software architecture, relating \emph{publishers}
	who send messages, and receivers called \emph{subscribers} (cf. for instance \cite{Eu:Th,Ha:Ap}). The main
	characteristics of this architecture are as follows. The publishers
	characterize messages according to classes/topics but they do not know whether
	there is any subscriber who is interested in a concrete topic. Subscribers, on
	the other hand, express their interest in one or more topics and receive all messages published to the topics to which they subscribe. All
	subscribers which are subscribed to a topic will receive the same messages  in case such messages exist (Figure \ref{pub-sub}). 
		A main problem in Publish/Subscribe architecture refers to the priority according to which a topic receives messages from the publishers as well as the priority according to which a subscriber receives messages from several topics (cf. for instance \cite{Ba:Dy,Ba:On,Zh:To}).
		 		 Recent applications in IoT framework as well as in cloud platforms incorporate Publish/Subscribe architecture  (cf. for instance \cite{Ol:AP,Pa:Pu, Ya:Pr}).  	
 	
 	In the sequel, we develop a \emph{wPCL} formula for the Publish/Subscribe architecture where, according to the semiring used, the weights represent in a natural way the aforementioned priorities. More precisely, we assign weights, describing priorities, to interactions among publishers and topics, and to interactions among topics and subscribers.  We consider three types
	of components namely, publishers, topics, and subscribers denoted by the
	letters $P,T,S$, respectively. Component $P$ has one port $p$, $T$ has two ports
	$t_{1}$ and $t_{2}$, and $S$ has the port $s$.  
		\begin{figure} 
			\begin{center}
			  	\begin{tikzpicture}
			  
			  \draw  (-0.5,0) rectangle  (1,1) ;
			  \node at (0.25,0.8) {\tiny $P_1$};
			  \draw  (0.5,0.3) rectangle (1,0.6);
			  \node at (0.8,0.45) {\tiny $p_1$};
			  \draw[fill] (1,0.45) circle [radius=1.5pt];

			  \draw  (-0.5,-2) rectangle  (1,-1) ;
			  \node at (0.25,-1.2) {\tiny $P_2$};
			  \draw  (0.5,-1.7) rectangle (1,-1.4);
			  \node at (0.8,-1.55) {\tiny $p_2$};
			  \draw[fill] (1,-1.55) circle [radius=1.5pt];

			  \draw[dashed] (0.25,-2.3)--(0.25,-3);
			  
			  \draw  (-0.5,-3.5) rectangle  (1,-4.5) ;
			  \node at (0.25,-3.7) {\tiny $P_n$};
			  \draw  (0.5,-4.2) rectangle (1,-3.9);
			  \node at (0.8,-4.05) {\tiny $p_n$};
			  \draw[fill] (1,-4.05) circle [radius=1.5pt];

			  \draw (4,-0.2) ellipse (0.9cm and 0.5cm);
			  \node at (4,-0.2) {\tiny $T_1$};
			  \draw (3.35,-0.2) ellipse (0.25cm and 0.25cm);
			  \node at (3.35,-0.2) {\tiny $t_{11}$};
			  \draw (4.65,-0.2) ellipse (0.25cm and 0.25cm);
			  \node at (4.65,-0.2) {\tiny $t_{12}$};
			  \draw[fill] (3.1,-0.2) circle [radius=1.5pt];
			  \draw[fill] (4.9,-0.2) circle [radius=1.5pt];

			  \draw (4,-1.6) ellipse (0.9cm and 0.5cm);
			  \node at (4,-1.6) {\tiny $T_2$};
			  \draw (3.35,-1.6) ellipse (0.25cm and 0.25cm);
			  \node at (3.35,-1.6) {\tiny $t_{21}$};
			  \draw (4.65,-1.6) ellipse (0.25cm and 0.25cm);
			  \node at (4.65,-1.6) {\tiny $t_{22}$};
			  \draw[fill] (3.1,-1.6) circle [radius=1.5pt];
			  \draw[fill] (4.9,-1.6) circle [radius=1.5pt];

			  \draw[dashed] (4,-2.5)--(4,-3);
			  
			  \draw (4,-3.8) ellipse (0.9cm and 0.5cm);
			  \node at (4,-3.8) {\tiny $T_m$};
			  \draw (3.38,-3.8) ellipse (0.3cm and 0.28cm);
			  \node at (3.35,-3.8) {\tiny $\ \ t_{m1} $};
			  \draw (4.6,-3.8) ellipse (0.3cm and 0.28cm);
			  \node at (4.65,-3.8) {\tiny $t_{m2} \ $};
			  \draw[fill] (3.1,-3.8) circle [radius=1.5pt];
			  \draw[fill] (4.9,-3.8) circle [radius=1.5pt];

			  \draw  (7.2,0) rectangle  (8.7,1) ;
			  \node at (7.95,0.8) {\tiny $S_1$};
			  \draw  (7.2,0.3) rectangle (7.6,0.6);
			  \node at (7.4,0.425) {\tiny $s_1$};
			  \draw[fill] (7.2,0.425) circle [radius=1.5pt];
			  
			  \draw  (7.2,-1.5) rectangle  (8.7,-0.5) ;
			  \node at (7.95,-0.7) {\tiny $S_2$};
			  \draw  (7.2,-1.2) rectangle (7.6,-0.9);
			  \node at (7.4,-1.075) {\tiny $s_2$};
			  \draw[fill] (7.2,-1.075) circle [radius=1.5pt];
			  
			  \draw  (7.2,-3) rectangle  (8.7,-2) ;
			  \node at (7.95,-2.2) {\tiny $S_3$};
			  \draw  (7.2,-2.7) rectangle (7.6,-2.4);
			  \node at (7.4,-2.575) {\tiny $s_3$};
			  \draw[fill] (7.2,-2.575) circle [radius=1.5pt];
			  
			  \draw[dashed] (7.95,-3.2)--(7.95,-3.7);
			  
			  \draw  (7.2,-5) rectangle  (8.7,-4) ;
			  \node at (7.95,-4.2) {\tiny $S_r$};
			  \draw  (7.2,-4.7) rectangle (7.6,-4.4);
			  \node at (7.4,-4.575) {\tiny $s_r$};
			  \draw[fill] (7.2,-4.575) circle [radius=1.5pt];

			  \draw (1,0.45)--(3.1,-0.2);
			  \node at (2.2,0.35) { \footnotesize $w_{11}$};
			  \draw  (1,0.45)--(3.1,-3.8);
			  \node at (2.25,-1.3) { \footnotesize $w_{1m}$};
			  
			  \draw (4.9,-0.2)--(7.2,0.425);
			  \node at (5.8,0.3) { \footnotesize $k_{11}$};

			  \draw (4.9,-1.6)--(7.2,0.425);
			  \node at (5.8,-0.45) { \footnotesize $k_{12}$};

			  \draw (4.9,-3.8)--(7.2,0.425);			
			  \node at (5.8,-1.5) { \footnotesize $k_{1m}$};

			  \end{tikzpicture}  
				
			\end{center}
			\caption{Weighted Publish/Subscribe architecture.}
			\label{pub-sub}
		\end{figure}
	
		\noindent   In our example we assume  two publisher components $P_1, P_2$ with ports $p_1, p_2$ respectively, four subscriber components $S_1, S_2, S_3,S_4$ with ports $s_1,s_2,s_3, s_4$ respectively, and three topic components $T_1, T_2, T_3$ with sets of ports $\{t_{11}, t_{12}\}$, $\{t_{21}, t_{22}\}$, and $\{t_{31},t_{32}\}$ respectively. Therefore, the set $P$ of all ports is given by  $P=\{p_1, p_2, s_1, s_2, s_3, s_4, t_{11}, t_{12}, t_{21}, t_{22}, t_{31}, t_{32}  \}$. For every $i\in \{1,2,3,4\}$, $j\in \{1,2,3\}$, and $l \in\{1,2\}$ we denote by $k_{ij} \in K $ the weight of the interaction among $S_i$ and $T_j$, i.e., the priority that the subscriber $S_i$ assigns to the receivement of a message from $T_j$, and by $w_{lj} \in K$ the weight of the interaction among $P_l$ and $T_j$, i.e., the priority that the topic $T_j$ assigns to the receivement of a message from $P_l$. The \emph{PIL} formula for the interaction between a publisher $P_l$ and a topic $T_j$, for every $l\in \{1,2\}$ and $j \in \{1,2,3\}$, is given by $\phi_{pt}(p_l,t_{j1})= m_{\{p_l, t_{j1}\}}$ 
	and the \emph{wPIL} formula characterizing this interaction with its corresponding weight is determined by $\varphi_{pt}(p_l,t_{j1})=w_{lj}\otimes \phi_{pt}(p_l,t_{j1}).$ 
	Hence, the \emph{wPCL} formula 
	$$ \zeta_j =  \varphi_{pt}(p_1,t_{j1})  \oplus  \varphi_{pt}(p_2,t_{j1})$$
	for $j \in \{1,2,3\}$, describes the weighted interactions of the topic $T_j$ with the publishers $P_1$ and $P_2$
	and thus 
	$$ \zeta_{t_j} = \sim( \varphi_{pt}(p_1,t_{j1})  \oplus  \varphi_{pt}(p_2,t_{j1}))$$
	refers to the weighted part of the Publish/Subscribe architecture among the concrete topic $T_j$ and the publishers $P_1$ and $P_2$.
	
\noindent	Next we build \emph{wPCL} formulas for the description of the interactions among subscribers and topics. More precisely, for every $i\in \{1,2,3,4\}$ and $j\in \{1,2,3\}$, the \emph{PIL} formula $\phi_{st} (s_i, t_{j2}) = m_{\{s_i, t_{j2}\}} $ 
	describes the interaction between the subscriber $S_i$ and the topic $T_j$, and the \emph{wPIL} formula $ \varphi_{st} (s_i, t_{j2}) = k_{ij} \otimes \phi_{st} (s_i, t_{j2})  $ 
	formalizes this interaction with its corresponding weight. Then, the \emph{wPCL} formula 
		\begin{align*}
		\zeta_{s_i} & = \sim(\varphi_{st}(s_i,t_{12})  \oplus  \varphi_{st}(s_i,t_{22})  \oplus \varphi_{st}(s_i,t_{32}) ) 
			\end{align*}
 characterizes the weighted interactions of subscriber $S_i$ with topics $T_1$, $T_2$, and $T_3$.
	
\noindent Finally, for every $i \in \{1,2,3,4\}$, we consider the \emph{wPCL} formula	
	\begin{align*}
	\zeta^{(tp)}_{s_i} & = \sim(( \varphi_{st}(s_i,t_{12}) \uplus \zeta_1) \oplus ( \varphi_{st}(s_i,t_{22}) \uplus \zeta_2)   \oplus (\varphi_{st}(s_i,t_{32}) \uplus \zeta_3)) 
		\end{align*}
	which describes the behavior of subscriber $S_i$ with publishers $P_1, P_2$ and topics $T_1,T_2,T_3$.
	
\noindent	For instance, let us consider the subscriber $S_1$, the set of interactions
	\begin{align*}
	\gamma  = & \{ \{p_1, t_{11}\},\{p_2, t_{11}\},\{p_1, t_{21}\}, \{p_2, t_{21}\}, \{p_1, t_{31}\},\{p_2, t_{31}\}, \{s_1, t_{12}\}, \{s_1, t_{22}\}, \{s_1, t_{32}\}\}
	\end{align*}
	and $K=\mathbb{R}_{\max}$. Then, the value $	\left\Vert \zeta^{(tp)}_{s_1} \right\Vert(\gamma) $
	represents the maximum priority with which the subscriber $S_1$ will receive a message. 	
	Similarly, for $K=\mathbb{R}_{\min}$,  the value $ \left\Vert \zeta^{(tp)}_{s_1} \right\Vert(\gamma) $
 corresponds to the minimum priority with which $S_1$ receives a message, whereas 
if  $K$ is the Viterbi semiring, it assigns the maximum probability for the receivement of the topics by $S_1$. 
\end{example}

\section{A full normal form for wPCL formulas}

In the present section, we show that for every wPCL formula $\zeta$ 
over $P$ and $K$ we can effectively compute an equivalent formula of a
special form. For this, we will use a corresponding result from \cite{Ma:Co}, namely
for every PCL formula $f$ over $P$\ we can effectively construct a unique equivalent PCL formula of
the form $\mathrm{true}$\footnote{This trivial case is not considered in \cite{Ma:Co} but we refer to this in order to show the correspondence with wPCL formulas $k$ for $k \in K$.} or $\bigsqcup\nolimits_{i\in I}\sum\nolimits_{i\in J_{i}}m_{i,j}$ which
is called \emph{full normal form }(cf. Theorem 4.43. in \cite{Ma:Co}). The
index sets $I$ and $J_{i}$, for every $i\in I$, are finite and $m_{i,j}$'s are
\emph{full monomials}, i.e., monomials of the form $\bigwedge_{p\in P_{+}}p\wedge
\bigwedge_{p\in P\_}\overline{p}$ with $P_{+}\cup P\_=P$ and $P_{+}\cap
P\_=\emptyset$. Here we firstly prove that the worst case run time for the construction of the full normal form  is doubly exponential, whereas the best case is exponential. Then we show that we can effectively build  a unique full
normal form for every wPCL formula $\zeta\in PCL(K,P)$ . Uniqueness is up to the
equivalence relation, and interestingly our algorithm requires the same time complexity as the corresponding one in the boolean case. Then we will use our translation result to show that the equivalence problem for wPCL formulas is decidable.

We start with the subsequent theorem which states the complexity of the construction of full normal forms for PCL formulas.  
\begin{theorem} \label{PCL-compl}
Let $P$ be a set of ports. Then, for every \emph{PCL} formula $f$ over $P$ we can effectively construct, in doubly exponential time, an equivalent \emph{PCL} formula $f'$ in full normal form. The best run time for the construction of $f'$ is exponential. Furthermore, $f'$ is unique up to equivalence relation. 
\end{theorem}

\begin{proof}
The construction of $f'$ and its uniqueness has been proved in Theorem 4.43. in \cite{Ma:Co}. Therefore, we only deal with the complexity result.  For this we recall, from \cite{Ma:Co}, the transformations  applied to $f$ for the construction of $f'$. More precisely, firstly the following rewriting rules are applied for the construction of a PCL formula $f''$ in normal form which is equivalent to $f$:

$$
\infer{\underset{i \in I}{\bigsqcup}g \wedge f_i} {g \wedge \underset{i \in I}{\bigsqcup}  f_i} \qquad \qquad
\infer{\underset{i \in I}{\bigsqcup}f_i + g} {g + \underset{i \in I}{\bigsqcup}f_i}  \qquad \qquad
\infer{\underset{i \in I}{\bigwedge}\lnot f_i }{\lnot \underset{i\in I}{\bigsqcup}f_i}  
$$

$$
\infer{\underset{\phi \in \Phi}{\bigsqcup}\overline{\phi} \ \sqcup \sim \left(\underset{\phi \in \Phi}{\bigwedge}\overline{\phi} \right)}{\lnot \underset{\phi \in \Phi}{\sum}\phi, \quad \text{all } \phi \text{ are interaction formulas}  }
$$

$$
\infer{\underset{\xi \in \Phi \cup \Psi} {\sum} \left(\xi \wedge \underset{(\phi,\psi)  \in  \Phi \times \Psi}{\bigvee}(\phi \wedge \psi) \right) }{\underset{\phi \in \Phi}{\sum}\phi \wedge \underset{\psi \in \Psi}{\sum} \psi, \quad \text{all } \phi\in \Phi \text{ and } \psi \in \Psi \text{ are interaction formulas}} 
$$

\noindent Then, every non-full monomial $m$ in $f''$ is transformed into a disjunction of  full monomials. This is done in exponential time, since all the subsets of the set $\overline{P}_{m}=P \setminus \{p \in P \mid  p \text{ or } \overline{p} \text{ occurs in } m\}$ are computed. Then using the rewriting rule 
$$
\infer{\underset{\emptyset \neq J \subseteq I}{\bigsqcup}\underset{j\in J}{\sum}f_i}{\underset{i\in I}{\bigvee}f_j}
$$

\noindent every disjunction of full monomials is transformed into a union of coalesced full monomials. Trivially, the application of the last rewriting rule on the derived disjunctions of full monomials needs again exponential time. Therefore, the application of the above rewriting rules requires at most doubly exponential time. This implies that the worst case run time for the computation of $f'$ is doubly exponential. On the other hand, if every monomial $m$ in $f''$ above is already full, then obviously the best run time for the computation of $f'$ is exponential, and this concludes our proof.
\end{proof}

Next we introduce the notion of full normal form for wPCL formulas.  
\begin{definition}
\label{def_fnf}A \emph{wPCL} formula $\zeta\in PCL(K,P)$ is said to be in
\emph{full normal form }if either 
\begin{itemize}
\item $\zeta=k$ with $k \in K$, or 
\item there are finite index sets $I$ and $J_{i}$ for
every $i\in I$, $k_{i}\in K$ for every $i\in I$, and full monomials $m_{i,j}$
for every $i\in I$ and $j\in J_{i}$ such that $\zeta=\bigoplus_{i\in I}\left(
k_{i}\otimes\sum_{j\in J_{i}}m_{i,j}\right)  $.
\end{itemize}
\end{definition}

By our definition above, for every full normal form we can construct an
equivalent one satisfying the following statements:

\begin{itemize}
\item[i)] $j\neq j^{\prime}$ implies $m_{i,j}\not \equiv m_{i,j^{\prime}}%
$\ for every $i\in I$, $j,j^{\prime}\in J_{i}$, and

\item[ii)] $i\neq i^{\prime}$ implies $\sum_{j\in J_{i}}m_{i,j}\not \equiv
\sum_{j\in J_{i^{\prime}}}m_{i^{\prime},j}$\ for every $i,i^{\prime}\in I$.
\end{itemize}

\noindent Indeed, for the first one, if $m_{i,j}\equiv m_{i,j^{\prime}}$\ for
some $j\neq j^{\prime}$, then\ since $m_{i,j},m_{i,j^{\prime}}$\ are
interaction formulas, by Proposition \ref{conj-inter}(i), we can replace the
coalescing $m_{i,j}+m_{i,j^{\prime}}$ with $m_{i,j}$. For (ii), let us assume
that $\sum_{j\in J_{i}}m_{i,j}\equiv\sum_{j\in J_{i^{\prime}}}m_{i^{\prime}%
,j}$\ for some $i\neq i^{\prime}$. Then, we can replace the sum $\left(
k_{i}\otimes\sum_{j\in J_{i}}m_{i,j}\right)  \oplus\left(  k_{i^{\prime}%
}\otimes\sum_{j\in J_{i^{\prime}}}m_{i^{\prime},j}\right)  $\ with the
equivalent one $\left(  k_{i}\oplus k_{i^{\prime}}\right)  \otimes\sum_{j\in
J_{i}}m_{i,j}$. Hence, in the sequel, we assume that every full normal form
satisfies Statements (i) and (ii).

We intend to show that for every wPCL formula $\zeta\in
PCL(K,P)$ we can effectively construct an equivalent wPCL
formula $\zeta^{\prime}\in PCL(K,P)$ in full normal form. Moreover,
$\zeta^{\prime}$ will be unique up to the equivalence relation\footnote{Since wPCL formulas are defined syntactically, we get for instance $\zeta_1 \uplus \zeta_2 \neq \zeta_2 \uplus \zeta_1$ whereas $\zeta_1 \uplus \zeta_2 \equiv \zeta_2 \uplus \zeta_1$.}. We shall need
a sequence of preliminary results. All index sets occurring in the
sequel are finite.

\begin{lemma}
\label{fullno3} Let $k_{1},k_{2}\in K$ and $\zeta_{1},\zeta_{2}\in PCL(K,P)$.
Then
\[
\left(  k_{1}\otimes\zeta_{1}\right)  \uplus\left(  k_{2}\otimes\zeta
_{2}\right)  \equiv\left(  k_{1}\otimes k_{2}\right)  \otimes\left(  \zeta
_{1}\uplus\zeta_{2}\right)  .
\]

\end{lemma}

\begin{proof}
For every $\gamma\in C(P)$ we have%
\begin{align*}
\left\Vert \left(  k_{1}\otimes\zeta_{1}\right)  \uplus\left(  k_{2}%
\otimes\zeta_{2}\right)  \right\Vert (\gamma)  &  =\bigoplus_{\gamma
=\gamma_{1}\cupdot \gamma_{2}}\left(  \left\Vert k_{1}\otimes\zeta_{1}\right\Vert
(\gamma_{1})\otimes\left\Vert k_{2}\otimes\zeta_{2}\right\Vert (\gamma
_{2})\right) \\
&  =\bigoplus_{\gamma=\gamma_{1}\cupdot \gamma_{2}}\left(  \left(  k_{1}%
\otimes\left\Vert \zeta_{1}\right\Vert (\gamma_{1})\right)  \otimes\left(
k_{2}\otimes\left\Vert \zeta_{2}\right\Vert (\gamma_{2})\right)  \right) \\
&  \equiv \bigoplus_{\gamma=\gamma_{1}\cupdot \gamma_{2}}\left(  k_{1}\otimes
k_{2}\otimes\left\Vert \zeta_{1}\right\Vert (\gamma_{1})\otimes\left\Vert
\zeta_{2}\right\Vert (\gamma_{2})\right) \\
&  \equiv \left(  k_{1}\otimes k_{2}\right)  \otimes \left( \bigoplus_{\gamma=\gamma_{1}%
\cupdot \gamma_{2}}\left(  \left\Vert \zeta_{1}\right\Vert (\gamma_{1}%
)\otimes\left\Vert \zeta_{2}\right\Vert (\gamma_{2})\right) \right) \\
&  =\left(  k_{1}\otimes k_{2}\right)  \otimes\left\Vert \zeta_{1}\uplus
\zeta_{2}\right\Vert (\gamma)\\
&  =\left\Vert \left(  k_{1}\otimes k_{2}\right)  \otimes\left(  \zeta
_{1}\uplus\zeta_{2}\right)  \right\Vert (\gamma),
\end{align*}
where the first equivalence holds by commutativity of $K$. 
\end{proof}

\begin{lemma}
\label{fullno5}Let $J$ be an index set and $m_{j}$ a full monomial for every
$j\in J$. Then, there exists a unique $\overline{\gamma}\in C(P)$ such that
for every $\gamma\in C(P)$ we have $\left\Vert \sum\nolimits_{j\in J}%
m_{j}\right\Vert (\gamma)=1$ if $\gamma=\overline{\gamma}$ and $\left\Vert
\sum\nolimits_{j\in J}m_{j}\right\Vert (\gamma)=0$ otherwise.
\end{lemma}

\begin{proof}
For every $m_{j}$, $j\in J$, there exists a unique interaction $a_{j}$ such
that $a_{j}\models_{i}m_{j}$. Then, it is straightforward to show that
$\overline{\gamma}=\{a_{j}\mid j\in J\}$ satisfies our claim.
\end{proof}

\begin{proposition}
\label{fullno6} Let $f$ be a $PCL$ formula over $P$. Then there exist finite
index sets $I$ and $J_{i}$ for every $i\in I$, and full monomials $m_{i,j}$
for every $i\in I$ and $j\in J_{i}$ such that
\[
f\equiv\bigoplus_{i\in I}\sum_{j\in J_{i}}m_{i,j}\equiv\bigoplus_{i\in
I}\left(  1\otimes\sum_{j\in J_{i}}m_{i,j}\right)  .
\]
\hide{
In particular
\[
\mathrm{true}\equiv\bigoplus_{\emptyset\not =N\subseteq M}\sum_{m\in N}m
\]
where $M$ is the set of all full monomials over $P$ such that for every
$m,m^{\prime}\in M$, if $m\neq m^{\prime}$, then $m\not \equiv m^{\prime}$.}
\end{proposition}

\begin{proof}
By Theorem 4.43. in \cite{Ma:Co} there exists a unique full normal form such that
$f\equiv\bigsqcup_{i\in I}\sum_{j\in J_{i}}m_{i,j}$ where $m_{i,j}$ are full
monomials. Using similar arguments as the ones after Definition \ref{def_fnf}
we can assume, without any loss, that the full normal form satisfies
Statements (i) and (ii).$\ $By Lemma \ref{fullno5}, for every $i\in I$ there
exists a unique $\overline{\gamma}_{i}\in C(P)$ such that for every $\gamma\in
C(P)$ we have $\left\Vert \sum_{j\in J_{i}}m_{i,j}\right\Vert (\gamma)=1$ if
$\gamma=\overline{\gamma}_{i}$ and $\left\Vert \sum_{j\in J_{i}}%
m_{i,j}\right\Vert (\gamma)=0$ otherwise. Then\ we have%
\begin{align*}
\left\Vert f\right\Vert (\gamma)  &  =\left\{
\begin{array}
[c]{ll}%
1 & \textnormal{ if }\gamma\models\bigsqcup_{i\in I}\sum_{j\in J_{i}}m_{i,j}\\
0 & \textnormal{ otherwise}%
\end{array}
\right. \\
&  =\left\{
\begin{array}
[c]{ll}%
1 & \textnormal{ if }\gamma=\overline{\gamma}_{i}\text{ for some }i\in I\\
0 & \textnormal{ otherwise}%
\end{array}
\right. \\
&  =\left\{
\begin{array}
[c]{ll}%
1 & \textnormal{ if }\left\Vert \sum_{j\in J_{i}}m_{i,j}\right\Vert (\gamma)=1\text{
for some }i\in I\\
0 & \textnormal{ otherwise}%
\end{array}
\right. \\
&  =\left\{
\begin{array}
[c]{ll}%
1 & \textnormal{ if }\left\Vert \bigoplus_{i\in I}\sum_{j\in J_{i}}m_{i,j}\right\Vert
(\gamma)=1\text{ }\\
0 & \textnormal{ otherwise}%
\end{array}
\right. \\
&  =\left\{
\begin{array}
[c]{ll}%
1 & \textnormal{ if }\left\Vert \bigoplus_{i\in I}\left(  1\otimes\sum_{j\in J_{i}%
}m_{i,j}\right)  \right\Vert (\gamma)=1\text{ }\\
0 & \textnormal{ otherwise}%
\end{array}
\right.
\end{align*}
for\ every $\gamma\in C(P)$, where the last but one equality holds by
Statement (ii). Hence we get $f\equiv\bigoplus_{i\in I}\sum_{j\in J_{i}%
}m_{i,j}\equiv\bigoplus_{i\in I}\left(  1\otimes\sum_{j\in J_{i}}%
m_{i,j}\right)  $ and we are done.
\hide{Next, we trivially get $\mathrm{true}\equiv\bigsqcup_{\emptyset\not =N\subseteq M}%
\sum_{m\in N}m$. Moreover, for every $m\in M$\ there exists a unique
interaction $a_{m}\in I(P)$ satisfying $m$, hence for every nonempty set
$N\subseteq M$\ there exists a unique $\gamma_{N}=\{a_{m}\mid m\in N\}\in
C(P)$ satisfying $\sum_{m\in N}m$. This implies that $\bigsqcup_{\emptyset
\not =N\subseteq M}\sum_{m\in N}m\equiv\bigoplus_{\emptyset\not =N\subseteq
M}\sum_{m\in N}m$, and we are done.}
\end{proof}

\begin{lemma}
\label{fullno8} Let $m_{i},m_{j}^{\prime}$ be full monomials for every $i\in
I$ and $j\in J$. Then
\[
\left(  \underset{i\in I}{\sum}m_{i}\right)  \otimes\left(  \underset{j\in
J}{\sum}m_{j}^{\prime}\right)  \equiv\left\{
\begin{array}
[c]{ll}%
\underset{i\in I}{\sum}m_{i} & \textnormal{ if }\underset{i\in I}{\sum}m_{i}%
\equiv\underset{j\in J}{\sum}m_{j}^{\prime}\\
0 & \textnormal{ otherwise.}%
\end{array}
\right.
\]

\end{lemma}

\begin{proof}
By Lemma \ref{fullno5} there exist $\overline{\gamma},\overline{\gamma
}^{\prime}\in C(P)$ such that for every $\gamma\in C(P)$ we have $\left\Vert
\sum\nolimits_{i\in I}m_{i}\right\Vert(\gamma) =1$ if $\gamma=\overline
{\gamma}$ and $\left\Vert \sum\nolimits_{i\in I}m_{i}\right\Vert (\gamma)=0$
otherwise, and $\left\Vert \sum\nolimits_{j\in J}m_{j}^{\prime}\right\Vert
(\gamma)=1$ if $\gamma=\overline{\gamma}^{\prime}$ and $\left\Vert
\sum\nolimits_{j\in J}m_{j}^{\prime}\right\Vert (\gamma)=0$ otherwise.
Therefore, for every $\gamma\in C(P)$ we get
\begin{align*}
\left\Vert \left(  \underset{i\in I}{\sum}m_{i}\right)  \otimes\left(
\underset{j\in J}{\sum}m_{j}^{\prime}\right)  \right\Vert (\gamma)  &
=\left\Vert \underset{i\in I}{\sum}m_{i}\right\Vert (\gamma)\otimes\left\Vert
\underset{j\in J}{\sum}m_{j}^{\prime}\right\Vert (\gamma)\\
&  =\left\{
\begin{array}
[c]{ll}%
1\otimes1 & \textnormal{ if }\underset{i\in I}{\sum}m_{i}\equiv\underset{j\in J}%
{\sum}m_{j}^{\prime}\text{ and }\gamma=\overline{\gamma}=\overline{\gamma
}^{\prime}\\
0 & \textnormal{ otherwise}%
\end{array}
\right. \\
&  =\left\{
\begin{array}
[c]{ll}%
1 & \textnormal{ if }\underset{i\in I}{\sum}m_{i}\equiv\underset{j\in J}{\sum}%
m_{j}^{\prime}\text{ and }\gamma=\overline{\gamma}=\overline{\gamma}^{\prime
}\\
0 & \textnormal{ otherwise}%
\end{array}
\right. \\
&  =\left\{
\begin{array}
[c]{ll}%
\left\Vert \underset{i\in I}{\sum}m_{i}\right\Vert (\gamma) & \textnormal{ if
}\underset{i\in I}{\sum}m_{i}\equiv\underset{j\in J}{\sum}m_{j}^{\prime}\text{
and }\gamma=\overline{\gamma}=\overline{\gamma}^{\prime}\\
0 & \textnormal{ otherwise}%
\end{array}
\right.
\end{align*}
which concludes our claim.
\end{proof}

\begin{lemma}\label{full_coal} Let $m_{i},m_{j}^{\prime}$ be full monomials for every $i\in
I$ and $j\in J$. If $m_i \equiv m'_j$ for some $i \in I$ and $j \in J$, then 
\[
\left(  \underset{i\in I}{\sum}m_{i}\right)  \uplus \left(  \underset{j\in
J}{\sum}m_{j}^{\prime}\right)  \equiv 0.
\]
\end{lemma}
\begin{proof}
By Lemma \ref{fullno5} there exist $\overline{\gamma},\overline{\gamma
}^{\prime}\in C(P)$ such that for every $\gamma\in C(P)$ we have $\left\Vert
\sum\nolimits_{i\in I}m_{i}\right\Vert(\gamma) =1$ if $\gamma=\overline
{\gamma}$ and $\left\Vert \sum\nolimits_{i\in I}m_{i}\right\Vert (\gamma)=0$
otherwise, and $\left\Vert \sum\nolimits_{j\in J}m_{j}^{\prime}\right\Vert
(\gamma)=1$ if $\gamma=\overline{\gamma}^{\prime}$ and $\left\Vert
\sum\nolimits_{j\in J}m_{j}^{\prime}\right\Vert (\gamma)=0$ otherwise.
Therefore, if $\overline{\gamma} \cap  \overline{\gamma}' \neq \emptyset $, then by definition of the coalescing operator on wPCL formulas, we get 
$$
\left\Vert \left(  \underset{i\in I}{\sum}m_{i}\right)  \uplus \left(
\underset{j\in J}{\sum}m_{j}^{\prime}\right)  \right\Vert (\gamma) =0$$
for every $\gamma \in C(P)$, and this concludes our claim.
\end{proof}

\begin{theorem}
\label{thm_fnf}Let $K$ be a commutative semiring and $P$ a set of ports. Then,
for every \emph{wPCL} formula $\zeta\in PCL(K,P)$ we can effectively
construct an equivalent \emph{wPCL} formula $\zeta^{\prime}\in PCL(K,P)$ in
full normal form which is unique  up to the
equivalence relation. The worst case run time for the construction algorithm is doubly exponential and the best case is exponential. 
\end{theorem}
\begin{proof}
We prove our theorem by induction on the structure of wPCL
formulas $\zeta$\ over $P$ and $K$. If $\zeta =k$ with $k\in K$, then we have nothing to prove. Though, we shall need to use instead the equivalent full normal form $\zeta'=\bigoplus_{\emptyset\not =N\subseteq M}\left(
k\otimes\sum_{m\in N}m\right)$ where $M$ is the set of all full monomials over $P$ such that for every $m,m' \in M$ if $m\neq m'$, then $m \not \equiv m'$. Next let $\zeta=f$ be a PCL formula. Then, we conclude our claim by
Proposition \ref{fullno6}.

Assume now that $\zeta_{1},\zeta_{2}\in PCL(K,P)$ and let $\zeta_{1}^{\prime
}=\bigoplus_{i_{1}\in I_{1}}\left(  k_{i_{1}}\otimes\sum_{j_{1}\in J_{i_{1}}%
}m_{i_{1},j_{1}}\right)  $ and $\zeta_{2}^{\prime}=\bigoplus_{i_{2}\in I_{2}%
}\left(  k_{i_{2}}\otimes\sum_{j_{2}\in J_{i_{2}}}m_{i_{2},j_{2}}\right)  $
be their equivalent full normal forms, respectively. \\
Let firstly $\zeta=\zeta_{1}\oplus\zeta_{2}$ and assume that not both of $\zeta_1$, $\zeta_2$ are constants in $K$. We consider the formula
$\zeta_{1}^{\prime}\oplus\zeta_{2}^{\prime}$. If $\sum_{j_{1}\in J_{i_{1}}%
}m_{i_{1},j_{1}}\not \equiv \sum_{j_{2}\in J_{i_{2}}}m_{i_{2},j_{2}}$\ for
every $i_{1}\in I_{1}$\ and\ $i_{2}\in I_{2}$, then we set $\zeta^{\prime
}=\zeta_{1}^{\prime}\oplus\zeta_{2}^{\prime}$. If this is not the case, let
us assume that $\sum_{j_{1}\in J_{i_{1}^{\prime}}}m_{i_{1}^{\prime},j_{1}}%
\equiv\sum_{j_{2}\in J_{i_{2}^{\prime}}}m_{i_{2}^{\prime},j_{2}}$\ for some
$i_{1}^{\prime}\in I_{1}$\ and\ $i_{2}^{\prime}\in I_{2}$. Then we have
\begin{multline*}
\zeta_{1}^{\prime}\oplus\zeta_{2}^{\prime}\equiv\left(  \bigoplus_{i_{1}\in
I_{1}\setminus\{i_{1}^{\prime}\}}\left(  k_{i_{1}}\otimes\sum_{j_{1}\in
J_{i_{1}}}m_{i_{1},j_{1}}\right)  \right)  \oplus\left(  \bigoplus_{i_{2}\in
I_{2}\setminus\{i_{2}^{\prime}\}}\left(  k_{i_{2}}\otimes\sum_{j_{2}\in
J_{i_{2}}}m_{i_{2},j_{2}}\right)  \right) \\
\oplus\left(  (k_{i_{1}^{\prime}}\oplus k_{i_{2}^{\prime}})\otimes\sum
_{j_{1}\in J_{i_{1}^{\prime}}}m_{i_{1}^{\prime},j_{1}}\right)  .
\end{multline*}
We continue in the same way, and we conclude to a full normal form
$\zeta^{\prime}$, which, by construction, it is equivalent to $\zeta$. \\
If $\zeta=k_1$ and $\zeta_2=k_2$ with $k_1, k_2 \in K$, then we let $\zeta'=k_1 \oplus k_2$.

\noindent Next let $\zeta=\zeta_{1}\otimes\zeta_{2}$ and assume that neither $\zeta_1$ nor $\zeta_2$ is a constant in $K$. We set \newline$\xi=\bigoplus
_{i_{1}\in I_{1}}\bigoplus_{i_{2}\in I_{2}}\left(  k_{i_{1}}\otimes k_{i_{2}%
}\otimes\sum_{j_{1}\in J_{i_{1}}}m_{i_{1},j_{1}}\otimes\sum_{j_{2}\in
J_{i_{2}}}m_{i_{2},j_{2}}\right)  $ and we have%
\begin{align*}
\xi &  =\bigoplus_{i_{1}\in I_{1}}\bigoplus_{i_{2}\in I_{2}}\left(  k_{i_{1}%
}\otimes k_{i_{2}}\otimes\sum_{j_{1}\in J_{i_{1}}}m_{i_{1},j_{1}}\otimes
\sum_{j_{2}\in J_{i_{2}}}m_{i_{2},j_{2}}\right) \\
&  \equiv\bigoplus_{i_{1}\in I_{1}}\bigoplus_{i_{2}\in I_{2}}\left(  \left(
k_{i_{1}}\otimes\sum_{j_{1}\in J_{i_{1}}}m_{i_{1},j_{1}}\right)
\otimes\left(  k_{i_{2}}\otimes\sum_{j_{2}\in J_{i_{2}}}m_{i_{2},j_{2}%
}\right)  \right)  \\
&  \equiv\left(  \bigoplus_{i_{1}\in I_{1}}\left(  k_{i_{1}}\otimes\sum
_{j_{1}\in J_{i_{1}}}m_{i_{1},j_{1}}\right)  \right)  \otimes\left(
\bigoplus_{i_{2}\in I_{2}}\left(  k_{i_{2}}\otimes\sum_{j_{2}\in J_{i_{2}}%
}m_{i_{2},j_{2}}\right)  \right) \\
&  =\zeta_{1}^{\prime}\otimes\zeta_{2}^{\prime}  \equiv\zeta_{1}\otimes\zeta_{2}=\zeta
\end{align*}
where the second equivalence follows by the distributivity property
of $K$. \newline Now, we translate $\xi$ to an equivalent full normal form
$\zeta^{\prime}$\ as follows. By Lemma \ref{fullno8}, for every $i_{1}\in
I_{1}$, $i_{2}\in I_{2}$, we get $\sum_{j_{1}\in J_{i_{1}}}m_{i_{1},j_{1}%
}\otimes\sum_{j_{2}\in J_{i_{2}}}m_{i_{2},j_{2}}\equiv\sum_{j_{1}\in J_{i_{1}%
}}m_{i_{1},j_{1}}$\ if $\sum_{j_{1}\in J_{i_{1}}}m_{i_{1},j_{1}}\equiv
\sum_{j_{2}\in J_{i_{2}}}m_{i_{2},j_{2}}$, and $\sum_{j_{1}\in J_{i_{1}}%
}m_{i_{1},j_{1}}\otimes\sum_{j_{2}\in J_{i_{2}}}m_{i_{2},j_{2}}\equiv
0$\ otherwise. Hence, in the first case we replace $\sum_{j_{1}\in J_{i_{1}}%
}m_{i_{1},j_{1}}\otimes\sum_{j_{2}\in J_{i_{2}}}m_{i_{2},j_{2}}$ by
$\sum_{j_{1}\in J_{i_{1}}}m_{i_{1},j_{1}}$,\ whereas in the second case by
$0$. Obviously, $\zeta^{\prime}$ is the required full normal form. \\
If $\zeta_1=k$ (resp. $\zeta_2=k$) with $k \in K$, then we set $\zeta'= \bigoplus_{i_{2}\in I_{2}%
}\left( k\otimes k_{i_{2}}\otimes\sum_{j_{2}\in J_{i_{2}}}m_{i_{2},j_{2}}\right)  $ (resp. $\zeta'= \bigoplus_{i_{1}\in I_{1}%
}\left( k\otimes k_{i_{1}}\otimes\sum_{j_{1}\in J_{i_{1}}}m_{i_{1},j_{1}}\right)  $. If $\zeta=k_1$ and $\zeta_2=k_2$ with $k_1, k_2 \in K$, then we let $\zeta'=k_1 \otimes k_2$.  

\noindent Finally let $\zeta=\zeta_{1}\uplus\zeta_{2}$. We set \newline%
$\zeta'=\bigoplus_{i_{1}\in I_{1}}\bigoplus_{i_{2}\in I_{2}}\left(  k'_{i_{1}%
}\otimes k'_{i_{2}}\otimes\left(  \left(  \sum_{j_{1}\in J_{i_{1}}}%
m_{i_{1},j_{1}}\right)  +\left(  \sum_{j_{2}\in J_{i_{2}}}m_{i_{2},j_{2}%
}\right)  \right)  \right)  $ where $k'_{i_1}$ and $k'_{i_2}$ are defined for every $i_1 \in I_1$ and $i_2 \in I_2$ respectively, as follows. If $m_{i_1,j_1} \not\equiv m_{i_2,j_2} $ for every $j_1 \in J_{i_1}$ and $j_2 \in J_{i_2}$, then we set $k'_{i_1}=k_{i_1}$ and $k'_{i_2}=k_{i_2}$, otherwise we let $k'_{i_1}=k'_{i_2}=0$. Then we get
\begin{align*}
\zeta' &  =\bigoplus_{i_{1}\in I_{1}}\bigoplus_{i_{2}\in I_{2}}\left(  k'_{i_{1}%
}\otimes k'_{i_{2}}\otimes\left(  \left(  \sum_{j_{1}\in J_{i_{1}}}%
m_{i_{1},j_{1}}\right)  +\left(  \sum_{j_{2}\in J_{i_{2}}}m_{i_{2},j_{2}%
}\right)  \right)  \right) \\
&  \equiv\bigoplus_{i_{1}\in I_{1}}\bigoplus_{i_{2}\in I_{2}}\left(  k_{i_{1}%
}\otimes k_{i_{2}}\otimes\left(  \left(  \sum_{j_{1}\in J_{i_{1}}}%
m_{i_{1},j_{1}}\right)  \uplus\left(  \sum_{j_{2}\in J_{i_{2}}}m_{i_{2},j_{2}%
}\right)  \right)  \right) \\
&  \equiv\bigoplus_{i_{1}\in I_{1}}\bigoplus_{i_{2}\in I_{2}}\left(  \left(
k_{i_{1}}\otimes\sum_{j_{1}\in J_{i_{1}}}m_{i_{1},j_{1}}\right)  \uplus\left(
k_{i_{2}}\otimes\sum_{j_{2}\in J_{i_{2}}}m_{i_{2},j_{2}}\right)  \right) \\
&  \equiv\bigoplus_{i_{1}\in I_{1}}\left(  \left(  k_{i_{1}}\otimes\sum
_{j_{1}\in J_{i_{1}}}m_{i_{1},j_{1}}\right)  \uplus\left(  \bigoplus_{i_{2}\in
I_{2}}\left(  k_{i_{2}}\otimes\sum_{j_{2}\in J_{i_{2}}}m_{i_{2},j_{2}}\right)
\right)  \right) \\
&  \equiv\left(  \bigoplus_{i_{1}\in I_{1}}\left(  k_{i_{1}}\otimes\sum
_{j_{1}\in J_{i_{1}}}m_{i_{1},j_{1}}\right)  \right)  \uplus\left(
\bigoplus_{i_{2}\in I_{2}}\left(  k_{i_{2}}\otimes\sum_{j_{2}\in J_{i_{2}}%
}m_{i_{2},j_{2}}\right)  \right) \\
&  \equiv\zeta_{1}^{\prime}\uplus\zeta_{2}^{\prime}  \equiv\zeta_{1}\uplus\zeta_{2}=\zeta
\end{align*}
where the first equivalence holds by definition of $k'_{i_1}$ ($i_1 \in I_1$) and  $k'_{i_2}$ ($i_2 \in I_2$) and Lemma \ref{full_coal}. The second equivalence
follows by Lemma \ref{fullno3}, and the third and fourth ones by Proposition
\ref{wcoal_prop}(iv). By definition of $k'_{i_1}$ for every $i_1 \in I_1$ and  $k'_{i_2}$ for every $i_2 \in I_2$, we conclude that $\zeta'$ is the required full normal form satisfying Statements (i) and (ii). 

Therefore, we have shown that for every $\zeta\in PCL(K,P)$ we can construct
an equivalent $\zeta^{\prime}\in PCL(K,P)$ in full normal form. The uniqueness of $\zeta^{\prime}$, up to equivalence, is derived in a straightforward way using Statements (i) and (ii). 
It remains to prove our claim for the time complexity of the above presented algorithm for the construction of $\zeta'$. We should note that the input of the algorithm consists of the set $P$ of ports and the wPCL formula $\zeta$. If $\zeta =k$ with $k \in K$, and $\zeta$ is involved in an $\oplus$ or $\uplus$ operation, then we use the full normal form      $\zeta'=\bigoplus_{\emptyset\not =N\subseteq M}\left(
k\otimes\sum_{m\in N}m\right)$. The computation of $\zeta'$ requires a doubly exponential time since we compute firstly the set $M$ all full monomials over $P$ and then all  nonempty subsets of $M$.      
If $\zeta=f$ is a PCL formula, then we conclude our claim by Theorem \ref{PCL-compl}  and Proposition \ref{fullno6}. Next, by preserving the notations from the first part of our proof,  for the induction steps for $\oplus$, $\otimes$, and $\uplus$ operations, we note that the construction of $\zeta'$ by $\zeta_1'$ and $\zeta'_2$ is polynomial in every case.   
\end{proof}

Next we show that the equivalence problem for wPCL formulas is decidable in doubly exponential time.

\begin{theorem}
Let $K$ be a commutative semiring and $P$ a set of ports. Then, for every $\zeta,\xi \in PCL(K,P)$ the equality $\|\zeta\|=\|\xi\|$ is decidable in doubly exponential time. 
\end{theorem}

\begin{proof}
By Theorem \ref{thm_fnf} we can effectively construct, in doubly exponential time, wPCL formulas $\zeta',\xi'$ in full normal form such that $\|\zeta\|=\|\zeta'\|$ and $\|\xi\|=\|\xi'\|$. Let us assume that 
$\zeta'=\bigoplus_{i\in I}\left(k_{i}\otimes\sum_{j\in J_{i}}m_{i,j}\right)  $  and 
$\xi'=\bigoplus_{l\in L}\left(k'_{l}\otimes\sum_{r\in M_l}m'_{l,r}\right)  $ which moreover satisfy Statements (i) and (ii). Then, by Statement (ii) we get that $\|\zeta'\|=\|\xi'\|$ iff the following requirements (1)-(3) hold:
\begin{itemize}
\item[1)] $\mathrm{card}(I)=\mathrm{card}(L)$,
\item[2)] $\{k_i \mid i \in I  \}= \{k'_l \mid l\in L \}$, and
\item[3)] 
\begin{itemize}
\item[a)] if $\mathrm{card}(I)=  \mathrm{card}(\{k_i \mid i \in I  \})$, then $\sum_{j\in J_{i}}m_{i,j} \equiv \sum_{r\in M_l}m'_{l,r}$ for every $i \in I$ and $l \in L$ such that $k_i=k'_l$, 
\item[or]
\item[b)] if $\mathrm{card}(I)>\mathrm{card}(\{k_i \mid i \in I  \})$, then we get 
$\zeta' \equiv \bigoplus_{i'\in I'}\left(k_{i'}\otimes\bigsqcup_{i\in R_{i'}}\sum_{j\in J_{i}}m_{i,j}\right)  $ where $I'  \varsubsetneq I$,  $k_{i'}$'s ($i' \in I'$) are pairwise disjoint, and $R_{i'}$ ($i' \in I'$) is the set of all $i$ in $I$ such that $k_i=k_{i'}$. Similarly, we get
$\xi' \equiv \bigoplus_{l'\in L'}\left(k'_{l'}\otimes\bigsqcup_{l\in S_{l'}}\sum_{r\in M_l}m'_{l,r}\right)  $ 
where $L'  \varsubsetneq L$, $k'_{l'}$'s ($l' \in L'$) are pairwise disjoint, and $S_{l'}$ ($l' \in L'$) is the set of all $l$ in $L$ such that $k'_{l}=k'_{l'}$. Then $\bigsqcup_{i\in R_{i'}}\sum_{j\in J_{i}}m_{i,j} \equiv \bigsqcup_{l\in S_{l'}}\sum_{r\in M_l}m'_{l,r}$ for every $i' \in I'$ and $l' \in L'$ such that $k_{i'}=k'_{l'}$.
\end{itemize}
\end{itemize}
By Lemma \ref{fullno5} the decidability of equivalences in (3a) is reduced to decidabilty of equality of sets of interactions corresponding to full monomials, whereas the decidabilitty of equivalences in (3b) is reduced to the decidability of equality of sets whose elements are sets of interactions corresponding to full monomials. All the aforementioned full monomials, and hence their corresponding interaction sets,  have been computed in the process of the construction of full normal forms $\zeta'$ and $\xi'$ (cf. the proof of Theorem \ref{thm_fnf}). Therefore, the decidability for the required equalities and equivalences in (1)-(3) above requires at most polynomial time, and our proof is completed. 
\end{proof}

In the last part of this section, we wish to show that wPCL is sound. For this, we need firstly to introduce the notion of soundness  in the setting of wPCL. According to our best knowledge soundness has been defined only for multi-valued logics, with values in the bounded distributive lattice $[0,1]$ with the usual $\mathrm{max}$ and $\mathrm{min}$ operations (cf. \cite{Ha:Me}). Let $\zeta \in PCL(K,P)$ and $\Sigma=\{\zeta_{1},\ldots
,\zeta_{n}\}$ with $\zeta_1, \ldots, \zeta_n \in PCL(K,P)$. We say that
$\Sigma$ \emph{proves} $\zeta$ and write
$\Sigma\vdash\zeta$ if $\zeta$ is derived from the formulas in $\Sigma$, using
the axioms of PCL (cf. page 17 in \cite{Ma:Co}) and the equivalences of Propositions
\ref{wcoal_prop} and \ref{conj_over_coal}. Furthermore, we
write $\Sigma\models\zeta$ iff for every $\gamma \in \bigcap_{1 \leq i \leq n}\mathrm{supp}(\zeta_i)$ such that $\left\Vert\zeta_{1}\right\Vert(\gamma) = \ldots= \left\Vert\zeta_{n}\right\Vert(\gamma)$, then $\left\Vert\zeta\right\Vert(\gamma)=\left\Vert\zeta_{1}\right\Vert(\gamma)$. Then, we say that wPCL is \emph{sound} if $\Sigma
\vdash\zeta$ implies $\Sigma\models\zeta$.  In particular, if $ f_1, \ldots, f_n, f $ are PCL formulas over $P$, then we write $ \{f_1, \ldots, f_n \} \vdash f$ if $f$ is derived from formulas in $ \{f_1, \ldots, f_n \}$ using the axioms of PCL (cf. page 17 in \cite{Ma:Co}). Similarly, we write $ \{f_1, \ldots, f_n \} \models  f $ if for every $\gamma \in C(P)$ we get $\gamma \models f$ whenever $\gamma \models f_i$ for every $1\leq i \leq n$.

\begin{theorem}
Let $K$ be a commutative semiring and $P$ a set of ports. Then the \emph{wPCL} over $P$ and $K$ is sound.
\end{theorem}

\begin{proof}
If $\Sigma$ is
a set of PCL formulas and $f$ a PCL  formula over $P$, then
$\Sigma\vdash f$ implies $\Sigma\models f$ holds true since PCL
is sound (cf. \cite{Ma:Co}). If $\Sigma$ is a set of wPCL 
formulas and $\zeta$ a wPCL formula in $PCL(K,P)$ such that $\Sigma\vdash\zeta$, then
we get $\Sigma\models\zeta$ by Definition \ref{wPCL_sem} and Propositions
\ref{wcoal_prop} and \ref{conj_over_coal}. Therefore,
wPCL is sound, and we are done.
\end{proof}

\section*{Discussion}

\begin{figure}
	\begin{center}
		\includegraphics[height=4.5cm, width=14cm]{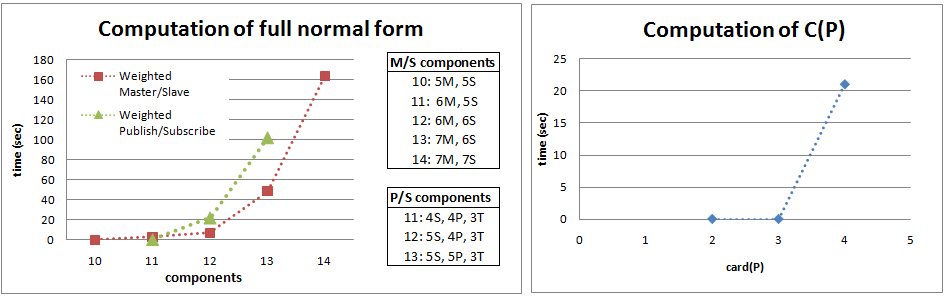}
	\end{center}
	\caption{Time required for computation of full normal forms and $C(P).$}
	\label{maude}
\end{figure}

\begin{itemize}
\item[a)] In our definition of interactions (cf. page \pageref{interaction}) over a set $P$ of ports we excluded, following \cite{Ma:Co}, the empty interaction as well as the empty set of interactions. Though the empty interaction adds no value in single architectures, it plays an important role in architectural composition (cf. for instance \cite{At:Ge,Bo:Lo}). More precisely, the empty set of interactions represents in a strict mathematical way the case where two architectures cannot be composed. It represents also a deadlock state of a component-based system. Nevertheless, we consider no empty interactions in wPCL because of the following reason. Our wPCL is based on PCL of \cite{Ma:Co} where the authors also consider no empty interactions. Clearly the empty interaction satifies only the PIL formula $\mathrm{false}$. By adding the empty interaction in PCL of \cite{Ma:Co} several properties do not hold any more. For instance the equivalences $f + \mathrm{false} \equiv \mathrm{false}$ (Proposition 4.4 in \cite{Ma:Co}) for $f \not \equiv \mathrm{false}$ and $\sim f_1 + \sim f_2 \equiv \sim f_1 \wedge \sim f_2$ (Proposition 4.27(2) in \cite{Ma:Co}) for $f_1 \not \equiv \mathrm{false}$ and $f_2 \equiv \mathrm{false}$. Specifically, the first one is used in the computation of the full normal form of a PCL formula (for instance in proof of Proposition 4.19 and in turn in proof of Proposition 4.35 in \cite{Ma:Co}). We conjecture that one can rebuilt the theory of PCL and of our wPCL by considering the empty interaction, but of course this is beyond the scope of this paper.            

\item[b)] In Theorem \ref{thm_fnf} we proved that for every wPCL formula $\zeta \in PCL(K,P)$ we can effectively construct an equivalent one $\zeta'$ in full normal form, i.e., $\zeta'= \bigoplus_{i\in I} \left( k_i\otimes \sum_{j\in J_i}m_{i,j} \right)$. Furthermore, we used this result to prove decidability of equivalence of wPCL formulas. However our method arises the following question: Why should one follow the technical constructions of Theorem \ref{thm_fnf} and not consider the ``obvious" construction of full formal form as $\zeta''=\bigoplus_{\gamma \in C(P)}\left( \left\Vert \zeta\right\Vert (\gamma) \otimes \sum_{a \in \gamma}m_a \right)$. As we have shown in Theorem \ref{thm_fnf} the worst case run time for the computation of $\zeta'$ is doubly exponential. On the other hand, the construction of $\zeta''$ requires ``always" doubly exponential time since it computes all $\gamma \in C(P)$. From the application point of view it is desirable to compute full normal forms in an automatic way \cite{Ma:Co}, and for this, the above difference for the time needed for $\zeta'$ and $\zeta''$ is of great importance. Following the methodology of \cite{Ma:Co}, we developed a code in Maude 2.7.1 rewriting system \cite{Maude} for the computation of full normal forms of weighted  Master/Slave and Publish/Subscribe architectures. The left diagram in Figure \ref{maude} presents the time needed for the construction of full normal forms of these architectures for several numbers of components. The right one shows the time needed to compute $C(P)$ for several numbers of $\mathrm{card}(P)$. One can easily see that the time required to compute $C(P)$ for $\mathrm{card}(P)=4$ is greater than  the the time required to compute the full normal form of weighted Master/Slave architecture with 6 Masters and 6 Slaves (i.e., 12 ports in total),  and it is the same with the time needed to compute the full normal of weighted Publish/Subscribe architecture with  5 Subscribers, 4 Publishers and 3 Topics (i.e., 12 ports in total). Our code run on a 64-bit Linux machine with a 2.10 GHz Intel i3-2310M CPU and 1 Gb memory limit.
\end{itemize}

\section*{Conclusion}
We introduced a wPCL over a set of ports $P$ and a commutative semiring $K$ and investigated several properties of the class of polynomials obtained as semantics of this logic. For some of that properties we required our semiring to be idempotent. In our wPCL we slightly modified the coalescing operation in comparison to its boolean counterpart. Definitely, this modification is not an essential restriction since every architecture described in the boolean setup is described as well in our weighted setup, as we showed in our examples. We proved that for every wPCL formula $\zeta$ we can effectively construct, in doubly exponential time,  an equivalent one $\zeta'$  in full normal form. This result implied the decidability of the equivalence problem for wPCL formulas. We defined a notion of soundness for wPCL and proved that it is sound. It should be noted that though PCL is proved to be complete \cite{Ma:Co}, for wPCL trivially this is not the case, at least with our definitions of $\vdash$ and $\models$. It is an open problem whether we can develop the theory of this paper by relaxing the commutativity property of the semiring $K$ and thus obtaining our results for a larger class of semirings. Furthermore, it should be very interesting to investigate the wPCL over more general weight structures which can describe further properties like average, limit inferior, limit superior, and discounting (cf. for instance \cite{Dr:Av}). Work on progress investigates the first- and second-order levels of weighted configuration logics which are motivated by applications to architecture styles with quantitative characteristics. 

\noindent \textbf{Acknowledgements.} We should like to express our gratitude to Joseph Sifakis, Simon Bliudze, and Anastasia Mavridou for useful discussions and  clarifications on \cite{Ma:Co}. We are also grateful to two anonymous reviewers for pointing out a mistake and fruitful suggestions which helped us to improve the presentation of the paper.

\end{document}